\documentclass[12pt,a4paper,titlepage,openany]{report}
\usepackage{zakljucna_FAMNIT_1_stopnja_MA_MEF_EN_2016}
\usepackage{verbatim}
\usepackage{diagbox}

\normalsize
\DeclareMathOperator{\pref}{\textit{pre}}
\DeclareMathOperator{\suff}{\textit{suf}}
\DeclareMathOperator{\size}{\textit{size}}
\DeclareMathOperator{\dt}{\textit{DT}}
\DeclareMathOperator{\hub}{\textit{hub}}

\begin{document}
\pagenumbering{Roman}
\pagestyle{empty}
\begin{center}
\noindent \large UNIVERZA NA PRIMORSKEM\\
\large FAKULTETA ZA MATEMATIKO, NARAVOSLOVJE IN\\
INFORMACIJSKE TEHNOLOGIJE

\normalsize
\vspace{5.5cm}
Zaklju\v cna naloga\\
(Final project paper)\\
\textbf{\large Od\v citljivost digrafov in dvodelnih grafov}\\
\normalsize
(Readability of digraphs and bipartite graphs)\\
\end{center}

\begin{flushleft}
\vspace{5cm}
\noindent Ime in priimek: Vladan Jovi\v ci\' c
\\
\noindent \v Studijski program: Matematika
\\
\noindent Mentor: izr.~prof.~dr.~Martin Milani\v c
\\
\noindent Somentor: prof.~dr.~Andrej Brodnik
\\
\end{flushleft}

\vspace{4cm}
\begin{center}
\large \textbf{Koper, julij 2016}
\end{center}
\newpage

\pagestyle{fancy}

\section*{Klju\v cna dokumentacijska informacija}

\medskip
\begin{center}
\fbox{\parbox{\linewidth}{
\vspace{0.2cm}
\noindent
Ime in PRIIMEK: Vladan JOVI\v CI\' C\vspace{0.5cm}\\
Naslov zaklju\v cne naloge: Od\v citljivost digrafov in dvodelnih grafov\vspace{0.5cm}\\
Kraj: Koper\vspace{0.5cm}\\
Leto: 2016\vspace{0.5cm}\\
\v Stevilo listov: 39\hspace{2cm} \v Stevilo slik: 13\hspace{2.6cm} \v Stevilo tabel: 2\hspace{2cm}\vspace{0.5cm}\\
\v Stevilo referenc: 8\vspace{0.5cm}\\
Mentor: izr. prof. dr. Martin Milani\v c\vspace{0.5cm}\\
Somentor: prof. dr. Andrej Brodnik\vspace{0.5cm}\\
Klju\v cne besede: od\v citljivost, graf prekrivanj, ozna\v cevalna funkcija, celo\v stevilski linearn program, razlo\v cljivost, dekompozicija, HUB-\v stevilo, dvodimenzionalne mre\v ze, toroidalne mre\v ze\vspace{0.5cm}\\
Math.~Subj.~Class.~(2010): 05C20, 05C62, 05C85, 05C75, 68W32, 90C10\vspace{0.5cm}\\
{\bf Izvle\v cek:}\\
V zaklju\v cni nalogi smo obravnavali invarianto grafov, imenovano od\v citljivost grafa. Motivacija za od\v citljivost prihaja iz bioinformatike. Grafi, ki se pojavijo v problemih sekvenciranja genoma, imajo majhno od\v citljivost, kar motivira \v studij grafov majhne od\v citljivosti. Predstavili smo algoritem Brage in Meidanisa, ki poka\v ze, da je parameter od\v citljivost dobro definiran in da je poljuben digraf mo\v zno predstaviti kot graf prekrivanj neke mno\v zice besed. Iz tega smo izpeljali zgornjo mejo za od\v citljivost.  Parameter od\v citljivost je mo\v c definirati tudi za dvodelne grafe; temu modelu je v zaklju\v cni nalogi posve\v cena posebna pozornost. Zahtevnost ra\v cunanja od\v citljivosti danega digrafa (ali dvodelnega grafa) \v se ni znana. Predstavili smo nov na\v cin za natan\v cno ra\v cunanje od\v citljivost s pomo\v cjo celo\v stevilskega linearnega programiranja. Obravnavali smo tudi dva pristopa za ra\v cunanje mej za od\v citljivost. Na koncu smo natan\v cno pora\v cunali od\v citljivost dvodimenzionalnih in toroidalnih mre\v z in predstavili polinomski algoritem za izra\v cun optimalne ozna\v cevalne funkcije dane dvodimenzionalne ali toroidalne mre\v ze.
\vspace{0.2cm}
}}
\end{center}

\newpage

\section*{Key words documentation}

\medskip

\begin{center}
\fbox{\parbox{\linewidth}{
\vspace{0.2cm}
\noindent
Name and SURNAME: Vladan JOVI\v CI\' C\vspace{0.5cm}\\
Title of final project paper: Readability of digraphs and bipartite graphs\vspace{0.5cm}\\
Place: Koper\vspace{0.5cm}\\
Year: 2016\vspace{0.5cm}\\
Number of pages: 39\hspace{1.6cm} Number of figures: 13\hspace{2.2cm} Number of tables: 2\vspace{0.5cm}\\
Number of references: 8\vspace{0.5cm}\\
Mentor: Assoc.~Prof.~Martin~Milani\v c, PhD\vspace{0.5cm}\\
Co-Mentor: Prof.~Andrej~Brodnik, PhD\vspace{0.5cm}\\
Keywords: readability, overlap graph, labeling, integer linear program, distinctness, decomposition, HUB-number, two-dimensional grid graphs, toroidal grid graphs\vspace{0.5cm}\\
Math.~Subj.~Class.~(2010): 05C20, 05C62, 05C85, 05C75, 68W32, 90C10\vspace{0.5cm}\\
{\bf Abstract:}\\
In the final project paper we consider a graph parameter called readability. Motivation for readability comes from bioinformatics applications. Graphs arising in problems related to genome sequencing are of small readability, which motivates the study of graphs of small readability. We present an algorithm due to Braga and Meidanis, which shows that every digraph is isomorphic to the overlap graph of some set of strings. An upper bound on readability is derived from the algorithm. The readability parameter can also be defined for bipartite graphs; in the final project paper special emphasis is given to the bipartite model. The complexity of computing the readability of a given digraph (or of a given bipartite graph) is unknown. A way for the exact computation of readability is presented using Integer Linear Programming. We also present two approaches for computing upper and lower bounds for readability due to Chikhi at al. Finally, the readability is computed exactly for toroidal and two-dimensional grid graphs and a polynomial time algorithm for constructing an optimal overlap labeling of a given two-dimensional or toroidal grid graph is presented.
\vspace{0.2cm}
}}
\end{center}


\newpage
\section*{Acknowledgement}

I would like to express my very great appreciation to my mentor and professor Martin Milani\v c for his guidance and suggestions during the planning, development and writing of this final project paper.

I would also like to thank my co-mentor and professor Andrej Brodnik for given advice and support during the writing of this paper.

My special thanks are extended to UP FAMNIT for giving me support during my study.

Finally, I wish to thank my parents and sister for their support and encouragement throughout my education.

\newpage

\tableofcontents
\addtocontents{toc}{\protect\thispagestyle{fancy}}
\newpage
\listoftables
\addtocontents{lot}{\protect\thispagestyle{fancy}}
\newpage
\listoffigures
\addtocontents{lof}{\protect\thispagestyle{fancy}}
\newpage

\chapter*{List of Abbreviations}
\thispagestyle{fancyplain}
\begin{longtable}{@{}p{1cm}@{}p{\dimexpr\textwidth-1cm\relax}@{}}
\nomenclature{{\it i.e.}}{that is}
\nomenclature{{\it e.g.}}{for example}
\end{longtable}
\newpage

\normalsize



\chapter{Introduction}
\thispagestyle{fancy}
\pagenumbering{arabic}

\section{Preliminaries}

In this final project we study a graph parameter called readability. Suppose we are given a finite set of finite strings $C$ over some alphabet $\Sigma$. Let $s_1, s_2 \in C$. For a positive integer $k$ such that $k \leq \min\{length(s_1), length(s_2)\}$ we say that $s_1$ overlaps $s_2$ by $k$ if the suffix of $s_1$ of length $k$ equals the prefix of $s_2$ of length $k$. Denote by $ov(s_1,s_2)$ the minimum $k$ such that $s_1$ overlaps $s_2$ and set $ov(s_1,s_2) = 0$ if $s_1$ does not overlap $s_2$. Given a set of strings $C$, we can construct the following directed graph: each string represents one vertex and there is directed edge between two vertices $u$ and $v$ if and only if $ov(u,v) > 0$. The graph obtained this way is called the \textit{overlap graph} of set $C$. Clearly this construction can be done in polynomial time. Now consider the reverse problem: given a directed graph $G$, find a set of strings $C$ such that $G$ is its overlap graph. The smallest integer $k$ for which there is a set of strings $C$ such that the length of each string in $C$ is at most $k$ and $G$ is the overlap graph of $C$ is called the readability of $G$ and denoted by $r(G)$. Note that the size of the alphabet of strings in $C$ is unrestricted. 

In this paper we will review the algorithm for building a set of strings for a given overlap graph presented in \cite{BragaMeidanis} and consequently we will show that readability is well defined and derive a (weak) upper bound (Chapter 2). In Chapters 3 and 4 we present an ILP formulation for the exact computation of readability, firstly for balanced bipartite graphs (for which parameter will be introduced in Chapter 3) and then for digraphs. We also present lower and upper bounds on readability in Chapter 5, given by \cite{MilanicRayanMedvedev}. In the last chapter, we explore the readability of graphs known as \textit{two-dimensional grid graphs} and \text{toroidal grid graphs}.

Motivation for the readability comes from bioinformatics applications. For example, well known graphs in bioinformatics are overlap graphs where each vertex represents some DNA sequence and two vertices are adjacent if and only if there is an overlap between corresponding sequences. There are several problems that have been studied on such graphs, e.g., the \textit{Minimum s-Walk Problem} and the \textit{De Bruijn Superwalk Problem}~\cite{deBruijn}. The problem of constructing a set of strings of a given (overlap) graph also has an application showing that a certain family of problems, a variation of the so-called Minimum Contig Problems, are NP-hard~\cite{BragaMeidanis}.

\section{Basic definitions and notations}
We study readability of finite graphs that can be directed or undirected. We will denote by $\pref (s,l)$ the prefix of string $s$ of length $l$. Similarly, $\suff (s,l)$ will denote the suffix of $s$ of length $l$. A \textit{labeling} $\ell$ of a graph $G$ is assignment of a string to each vertex of $G$. The \textit{length} of a labeling $\ell$ is denoted by $len(\ell)$ and defined as the maximum length of a string in the image of $\ell$. An \textit{overlap labeling} of a digraph $D$ is a labeling $\ell$ of $D$ such that for all pairs $u,v$ of vertices of $D$, the pair $(u,v)$ is an arc of $D$ if and only if $ov(\ell(u),\ell(v))>0$. We will denote the $i$th character of the string assigned to some vertex $u$ by $u(i)$.

\begin{definition}
Let $D$ be a digraph. The readability of $D$, denoted by $r(D)$ is the minimum positive integer $k$ such that there exists an injective overlap labeling of $D$ of length $k$.
\end{definition}

We will now introduce some standard notions in graph theory used to study the readability.

\begin{definition}
The {\it chromatic index} $\chi '(G)$ of an undirected graph $G$ is the minimum number of colors needed to color the edges of $G$ such that no two distinct edges that share an endpoint have the same color.
\end{definition}
\begin{definition}
Given an undirected graph $G=(V,E)$, a set $M \subseteq E(G)$ is called a \emph{matching} if each vertex in $V(G)$ is incident to at most one edge from $M$.
\end{definition}

\begin{definition}
The \emph{disjoint union of graphs} $H_1$ and $H_2$ with disjoint vertex sets is the graph $H = H_1 + H_2$ with $V(H) = V(H_1) \cup V(H_2)$ and $E(H) = E(H_1) \cup E(H_2)$.
\end{definition}

We denote the maximum degree of an undirected graph $G$ with $\Delta(G)$. In the case of a digraph $G$, $\Delta^+(G)$ denotes the maximum out-degree and $\Delta^-(G)$ denotes the maximum in-degree. All bipartite graphs considered in this final project paper will be assumed given with a bipartition of their vertex set into two independent sets (called parts) and thus denoted by $G=(S\cup T, E)$, where $S$ and $T$ are parts of $G$.

\begin{definition}
A bipartite graph $G=(S\cup T, E)$ with parts $S$ and $T$ is called a \emph{biclique} if for any two vertices $u \in S$ and $v \in T$ we have $uv \in E(G)$.
\end{definition}

\chapter{An upper bound on readability}

As stated in introduction, the readability of a digraph is well defined. To be able to prove this and give an upper bound for the readability, we will define the notions of a {\it directed matching} and a {\it directed edge coloring}.
\begin{definition}
A {\it directed matching} of a directed graph $G$ is a set $M \subseteq E(G)$ such that for any two distinct arcs $u_1v_1 \in M$ and $u_2v_2 \in M$ we have $u_1 \neq u_2$ and $v_1 \neq v_2$ i.e., the tails of the edges have to be different and the heads have to be different.
\end{definition}
\begin{definition}
Let $L = \{M_1, M_2, \dots, M_k\}$ be a a collection of pairwise disjoint directed matchings of a graph $G$. If $L$ covers all the edges of $G$, i.e., each edge belongs to $M_i$ for some $i \in \{1, \dots, k\}$, then $L$ is said to be a \textit{directed edge coloring}.
\end{definition}
\begin{theorem}[Braga and Meidanis~\cite{BragaMeidanis}]\label{thm:wellD}
For an arbitrary directed graph $G = (V,A)$ there exists a labeling $\ell$ of $G$ such that $len(\ell) \leq 2^{p+1} -1$ where $p = \max\{\Delta^{+}(G), \Delta^{-}(G)\}$ and $G$ is the overlap graph of the corresponding set of strings.
\end{theorem}
We follow the proof of Braga and Meidanis~\cite{BragaMeidanis} and to prove this theorem, we will need the following two theorems.
\begin{theorem}[K\" onig's Line Coloring Theorem]\label{thm:Konig}
Every bipartite graph $G$ satisfies $\chi'(G) = \Delta(G)$.
\end{theorem}
A proof of Theorem~\ref{thm:Konig} can be found in~\cite{Diestel}.
\begin{theorem}\label{thm:dec}
For every directed graph $G$ there is a directed edge coloring $L = \{M_1, M_2, \dots, M_p\}$ where $p =\ max\{\Delta^{+}(G), \Delta^{-}(G)\}$.
\end{theorem}
\begin{proof}
Given a directed graph $G$ we construct a bipartite graph $H$ as follows:
\begin{itemize}
\item for every $v \in V(G)$ add two vertices $v'$ and $v''$ to $V(H)$
\item for every $e=uv \in E(G)$ add the edge $e' = u'v''$ to $E(H)$
\end{itemize}
It is obvious that $H$ is bipartite with $|V(H)| = 2|V(G)|$ and $|E(H)| = |E(G)|$. By Theorem~\ref{thm:Konig} we have that the minimum number of colors needed to color the edges of $H$ is $\Delta(H) = \max\{\Delta^{+}(G), \Delta^{-}(G)\} = p$. Observe that a set of edges colored by the same color is a matching of graph $H$. Let $L' = \{M_1', M_2', \dots, M_p'\}$ be a collection of pairwise disjoint matchings of $H$ given by an optimal edge coloring.
Define $M_i = \{uv \mid u'v'' \in M_i'\}$.
Since $M_i'$ is matching of $H$, we have that $M_i$ is a directed matching in $G$ and since matchings in $L'$ were pairwise disjoint we conclude that directed matchings in $L = \{M_1, M_2, \dots, M_p\}$ are also pairwise disjoint. Also, since the edge coloring assigns a unique color to each edge, we have that for each edge $e \in E(H)$ there exists some $i \in \{1, 2, \dots, p \}$ such that $e \in M_i'$. This means that $L$ covers all edges of $G$. Thus, $L$ is a directed edge coloring of $G$ and $|L| = p$.
\end{proof}
Now we are ready to prove Theorem~\ref{thm:wellD}. We propose a constructive proof. The idea is to generate a string of the form $\pref (v).{\it Unique(v)}.\suff (v)$ for each vertex $v \in V(G)$ where $Unique(v)$ is a character not used before, i.e., not occurring in $\pref (v)$ nor $\suff (v)$ nor in any other string assigned to vertex $u \in V(G)$. This way, the number of different characters that appear in the strings assigned to vertices is bounded from below. In fact, algorithm produces labels on vertices using exactly $|V(G)| + |E(G)|$ characters in total.
\begin{algorithm}[h!]\label{algorithm11}
\KwIn{Directed graph $G$}
\KwOut{A set of strings $C = \{s_u : u \in V(G)\}$ whose overlap graph is $G$}
\caption{Construction of an overlap labeling of given digraph $G$}
{
    \For{$v \in V(G)$}
    {
        $\pref (v) = \epsilon$; \\
        $\suff (v) = \epsilon$;
    }
    {	
    Find a minimum directed edge coloring $L=\{M_1,M_2,\dots,M_p\}$ where $p=max\{\Delta^{+}(G), \Delta^{-}(G)\}$ \\
    }
    {$I$ = 0;} \\
    \For{$M_i \in L$}
    {
    	\For{$e=uv \in M_i$}
        {
        	$\pref (v) = \pref (v).I.\suff (u)$; \\
            $\suff (u) = \pref (v)$; \\
            $I = I + 1$; \\
        }
    }
    {$C = \emptyset$} \\
    \For{$v \in V(G)$}
    {
    	$s_v = \pref (v).I.\suff (v)$; \\
    	$C = C \cup \{ s_v \}$; \\
        $I = I + 1$;
    }
}
\Return{$C$;}
\end{algorithm}
Figure~\ref{fig:alg_exec1} provides an example of algorithm execution. Firstly, we calculate a minimum directed edge coloring $L = \{M_1, M_2\}$ of size $p = \max\{\Delta^+(D), \Delta^-(D)\} = 2$ and then we process the edges one by one. The maximal length of string after processing all edges and merging is $5$, showing that $r(D) \leq 5$
\bigskip
\bigskip
\begin{figure}[h!]
\begin{center}
\includegraphics[width=1\linewidth]{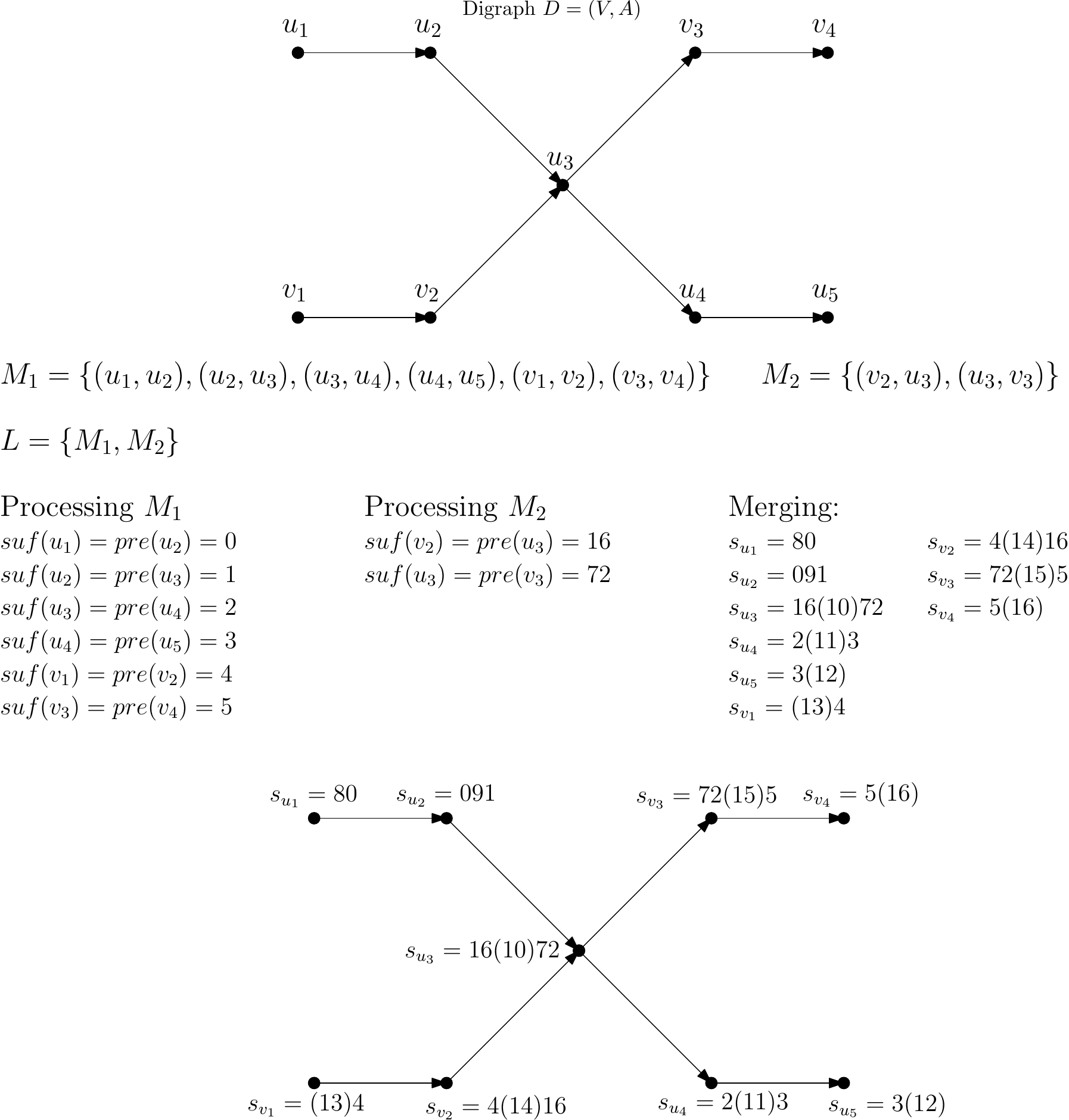}
\end{center}
\caption{Execution of Algorithm~\ref{algorithm11} on digraph $D$.}\label{fig:alg_exec1}
\end{figure}

Finally, to prove that the readability is well defined, we have to prove correctness of Algorithm~\ref{algorithm11}.
\begin{proof}[Proof of Theorem~\ref{thm:wellD}]
We have to prove few things:\begin{enumerate}
\item $uv \in E(G) \iff \exists i \in \mathbb{N}$ s.t. $\suff (s_u,i) = \pref (s_v,i)$ and
\item $\forall s_u \in C$, $|s_u| \leq 2^{p+1} - 1$.
\end{enumerate}
First, we will prove that if $e = uv \in E(G)$ then there exists a positive integer $i$ such that after the execution of algorithm $\suff (s_u,i) = \pref (s_v,i)$.
If $uv \in E(G)$ then $\exists i \in \{1,2,\dots,p\}$ s.t. $uv \in M_i$ since $L$ is a directed edge coloring of $G$.
Thus the edge is processed in steps $7,8,9$. In step $8$ an overlap of positive length is created between partial strings assigned to vertices $u$ and $v$. In further execution of the algorithm, we only extend $\suff (u)$ to the left and only extend $\pref (v)$ to the right, which means that this overlapping will exist at the end of the execution of the algorithm.

Now, we prove the opposite direction: if there is a positive integer $l$ s.t. $\suff (s_u,l) = \pref (s_v, l)$ then $uv \in E(G)$.
Let $i=d^{+}(u)$ and order the out neighborhood of $u$ as $N^{+}(u) = \{v_1, \dots, v_i\}$ according to the order in which the edges $uv_i$ are processed by the algorithm, meaning that $uv_1$ is first processed, $uv_2$ second and so on.
Denote by $I(u,v_i)$ the value of $I$ in the algorithm right before the edge $uv_i$ is processed. Observe that $I(u,v_i) > I(u,v_{i-1}) > \dots > I(u,v_1)$ which follows from the definition of the ordering of the vertices in the out neighborhood of $u$.
Since we chose $I(u,v_i)$ as a unique character, it appears exactly once in $\suff (u)$ after the execution of the algorithm.
Moreover, $I(u,v_i)$ is the largest character in $\suff (u)$.
Similarly, since we extend $\suff(u)$ only to the left, we conclude that for $1 \leq i' < i$, character $I(u,v_{i'})$ appears exactly once and  is the largest among all characters to the right of $I(u,v_{i'+1})$.
From the other side, let $j=d^{-}(v)$ and order neighborhood $N^{-}(v) = \{u_1, \dots, u_j\}$ according to the order in which the edges $u_iv$ are processed by the algorithm.
Similarly as above, one can conclude that largest integer in any prefix $\pref (v)$ is one of the unique characters, i.e. of the form $I(u_{j'},v)$ for some $j' \in \{1,\dots, j\}$. Now, if there is a positive integer $l$ such that $s = \suff (s_u,l) = \pref (s_v,l)$ then by the above, the largest character in $s$ is of the form $I(u,v_{i'})$ for some $i' \in \{1,2,\dots, i\}$ and of the form $I(u_{j'},v)$ for some $j' \in \{1,2,\dots, j\}$ and since there is an overlap, they must be the same. Thus, $I(u,v_{i'})=I(u_{j'},v) \implies u = u_{j'}$ and $v = v_{i'}$ which shows that $uv \in E(G)$.

Here, we will prove inductively that $\forall u \in V(G)$ and for $s_u \in C$ we have $|s_u| \leq 2^{p+1} -1$ where $s_u$ denotes a string assigned to vertex $u$. Denote by $size(i)$ for $i \in \{0,1,\dots,p\}$ the maximum length of a strings assigned to $\pref(u)$ and $\suff (u)$ for $u \in V(G)$ right after processing all edges from $M_i$. size(0) = 0 since before processing $M_1$ we have only empty strings.
During the processing the edges of $M_i$, the strings $\pref (v)$ and $\suff (v)$, where $v$ is head or tail of a edge, are modified at most once since $M_i$ is directed matching and by definition there is no vertex which is head for two edges and there is no vertex which is tail for two edges. Observe that during processing edge $uv \in M_i$ we only modify $\pref (v)$ by concatenating strings obtained in previous step $i-1$. Thus, we have $size(i) = 2\cdot size(i-1) + 1$. We use the above fact and induction on $i$ to prove $\size(i) \leq 2^i - 1$:
\begin{enumerate}
	\item (basic step) $i=0$: $0 = \size(0) \leq 2^0 - 1 = 0$.
	\item (induction step) $i \rightarrow i+1$: $\size(i+1) \leq 2\cdot \size(i) + 1 \leq 2\cdot (2^i - 1) + 1 = 2^{i+1} - 1$.
\end{enumerate}
By the above proof, we have $size(p) \leq 2^p -1$, which means that the maximum length of strings assigned to $\pref(u)$ or $\suff(u)$ after the execution of the algorithm is at most $2^p -1$ for every vertex $u \in V(G)$. Since at the end we concatenate $\pref (u)$ and $\suff (u)$ with additional letter, we have $|s_u| \leq 2\cdot (2^p -1) +1 = 2^{p+1} -1$.
\end{proof}

The running time of the Algorithm~\ref{algorithm11} is $O(2^p(n+m))$ where $m = |E(G)|$ and $n = |V(G)|$. First, in steps $1,2,3$ we initialize the values of $\pref (v)$ and $\suff (v)$. This is done in linear time in the number of vertices, that is $O(n)$. Then, we compute a directed edge coloring in step $4$. By~\cite{BipartiteEdgeColoring} this can be done in time $O(pm)$. The time needed for steps $6-10$ can be bounded by $\sum_{i=1}^{|L|}{|M_i|\cdot (2^{p+1}-1)}$. The factor $2^{p+1}-1$ appears because we have to iterate through all characters of strings assigned to the vertices and by previous proof $|s_u| \leq 2^{p+1}-1$. Thus, $\sum_{i=1}^{|L|}{|M_i|\cdot (2^{p+1}-1)} = (2^{p+1}-1)\cdot \sum_{i=1}^{|L|}{|M_i|} = (2^{p+1}-1)\cdot m \in O(2^pm)$. A similar reasoning applies to steps $12,\dots,15$, for which the running time is $\sum_{i=1}^{|V(G)|}{(2^{p+1}-1)} \in O(2^pn)$. Summing up everything, we get that the time complexity of the algorithm is given by $O(2^p(n+m))$.

As we can see, to have the algorithm working correctly we must have that the size of alphabet is at least $|V(G)|+|E(G)|$. Braga and Meidanis~\cite{BragaMeidanis} showed that given a directed graph $G$ and an alphabet $\Sigma'$ with at least two different symbols, one can compute a set $C'$ of strings over $\Sigma'$ such that the overlap graph of $C'$ is $G$ in two steps:
\begin{enumerate}
\item Using Algorithm~\ref{algorithm11} compute a set of strings $C$ written over alphabet $\Sigma$ s.t. $|\Sigma| = n+m$, where the overlap graph of $C$ is $G$.
\item Map each string of $C$ into another string to obtain a set $C'$.
\end{enumerate}
The maximum length of the string in $C'$ is bounded from above by the maximum length of strings in $C$ multiplied by a factor of $O(\log_{|\Sigma'|}(n+m))$.

At the expense of possibly increasing the alphabet size, the algorithm could be modified to produce strings of the same length. In the last steps during concatenation of $\pref(v)$ and $\suff(v)$ i.e., in step 13, instead of one unique character $I$, one can add arbitrary many unique characters between $\pref(v)$ and $\suff(v)$ to achieve desired length.

\chapter{Readability of bipartite graphs and an integer programming formulation}

\section{Readability of bipartite graphs}

In the first chapter we have seen how the readability is defined for directed graphs. Here, we will define and study the readability of bipartite graphs in order to try to understand behavior of this parameter of digraphs.
This is possible at least approximately (see Theorem~\ref{thm:rdb_assymp}).

Given two finite sets of finite strings $S_s, T_s$, the bipartite overlap graph of $(S_s,T_s)$ is the bipartite graph $G = (S\cup T, E)$ with parts $S = \{u_s : s \in S_s\}$ and $\{v_t : t \in T_s\}$ such that $u_sv_t \in E(G)$ if and only if $ov(s,t)>0$.
A \textit{labeling} of a bipartite graph $G$ is function $\ell$ assigning a string to each vertex of $G$ such that all string have the same length. The \textit{length of a labeling} of a bipartite graph is defined and denoted the same as the length of labeling of a digraph. An \textit{overlap labeling} of a bipartite graph $G$ is a labeling of $G$ such that for all $u\in S$ and $v\in T$, we have $uv \in E(G)$ if and only if $ov(\ell(u),\ell(v))>0$.

\begin{definition}
The \emph{readability} of a bipartite graph $G = (S\cup T, E)$ is the minimum non-negative integer $k$ such that there is an overlap labeling $\ell$ of $G$ with $len(\ell) = k$.
\end{definition}
Observe that we do not request the labeling to be injective, which is not the case for digraphs. Moreover, we only care about overlaps ``in one direction", from $S$ to $T$. This gives us more flexibility in the study of readability.
Also, the assumption that the strings assigned to vertices by a labeling are of the equal length is without loss of generality (which is not the case for digraphs, see Chapter 4).

Denote by $D_n$ the set of all digraphs with vertex set $[n]:=\{1,\ldots,n\}$ and by $B_{n\times n}$ the set of all balanced bipartite graphs with a copy of $[n]$ in each part of the bipartition. The next theorem shows that as long as we are interested in the readability approximately, we can focus our attention to balanced bipartite graphs instead of to digraphs.
\newline
\begin{theorem}[Chikhi at al.~\cite{MilanicRayanMedvedev}]\label{thm:rdb_assymp}
There exists a bijection $\phi : B_{n\times n} \to D_n$ such that for each graph $H \in B_{n\times n}$ and $G \in D_n$ with $\phi(H) = G$ the following holds: $r(G) < r(D) \leq 2\cdot r(G) + 1$.
\end{theorem}
A proof of Theorem~\ref{thm:rdb_assymp} can be found in~\cite{MilanicRayanMedvedev}.

\section{Variables of the ILP}
Suppose we are given a balanced bipartite graph $G = (S \cup T, E)$ and a positive integer $r$. Consider the following decision problem: does $G$ have readability $r(G) \leq r$? To answer this question, we will describe an integer linear program that has a feasible solution if and only if $r(G) \leq r$.

For every pair of vertices $u \in S$ and $v \in T$ and for each $i,j \in \{1,2,\dots, r\}$ we introduce a binary variable $x_{u,v,i,j}$ which is equal to $1$ if the $i$th character of the string assigned to vertex $u$ is equal to $j$th character of the string assigned to vertex $v$ and equal to $0$ if the above characters are different. Furthermore, for each edge $e = uv \in E(G)$ and for each $i \in \{1,2,\dots, r\}$ we define a binary variable $z_{e,i}$ which is equal to $1$ if there is a overlapping of the size $i$ between vertices $u$ and $v$, i.e., $x_{u,v,i-k+1,k} = 1$ $\forall k \in \{ 1,\dots,i \}$. The variable $z_{e,i}$ is equal to $0$ if the above is not true.
\section{Constraints of the ILP}
We add the constraints listed below.
\begin{itemize}
\item \textbf{Edge constraints:} for each edge $e = uv \in E(G)$, $u \in S$ and $v \in T$, add the following constraints:
$$
\sum_{i=1}^r{z_{e,i}} \geq 1,
$$
$$
\sum_{k=i}^{r}{x_{u,v,k,k-i+1}} \geq (r-i+1) \cdot z_{e,i} \quad \text{for all } i \in \{1,\dots, r\}.
$$
\item \textbf{Non-edge constraints:} for each $u \in S$ and $v \in T$ such that $uv \not \in E(G)$ add the following constraints:
$$
\sum_{k=i}^{r}{(1-x_{u,v,k,k-i+1})} \geq 1 \quad \text{for all } i \in \{1, \dots, r\}.
$$
\item \textbf{Transitivity constraints: } for each $u,w \in S$ such that $u \neq w$ and for each $v,q \in T$ such that $v \neq q$ and for each $i,j,k,l \in \{1,\dots, r\}$ add
$$
x_{u,v,i,j} + x_{w,v,k,j} + x_{w,q,k,l} - x_{u,q,i,l} \leq 2
$$
\end{itemize}
Constraints of the first type ensure that we will have an overlap between adjacent vertices. Constraints of the second type ensure that there is no overlap between non-adjacent vertices. Constraints of the third type (\textit{transitivity constraints}) we introduce to ensure the following implication: if $u(i) = v(j)$ and $v(j) = w(k)$ and $w(k) = q(l) \implies u(i) = q(l)$.
\section{Objective function of the ILP and an example}
Since we are looking only for a feasible solution the objective function can be arbitrary, for example the constant $0$. Now, we have completed the description of the integer linear program.
An example of the defined variables and constraints is given in Figure~\ref{fig:vars}.
\begin{figure}[h!]
\begin{center}
\includegraphics[width=1\linewidth]{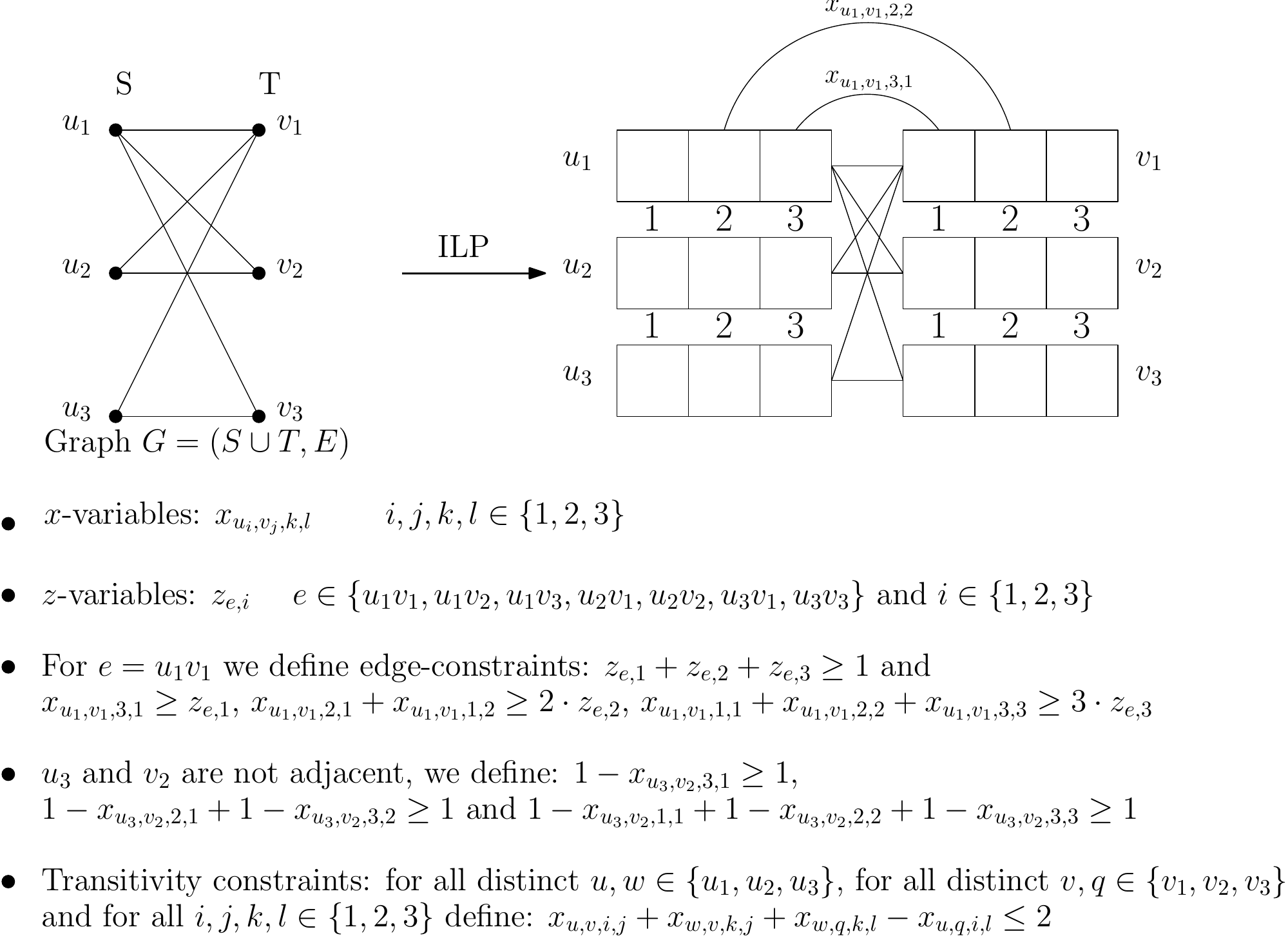}
\end{center}
\caption{Example of defined variables and constraints for the ILP with given graph $G = (S\cup T, E)$ and integer $r = 3$.}\label{fig:vars}
\end{figure}

\section{Proof of correctness}
\begin{theorem}\label{thm:model}
The above defined integer linear program has a feasible solution if and only if given graph is of the readability at most $r$.
\end{theorem}

\begin{proof}
($\impliedby$) Let $G = (S \cup T, E)$ be a balanced bipartite graph of readability at most $r$ and let $S_s,T_s$ be a sets of strings with $\forall s \in S_s\cup T_s$ $|s| \leq r$ such that the overlap graph of $(S_s,T_s)$ is isomorphic to $G$. We may assume without loss of generality that $|s| = r \quad \forall s \in S_s\cup T_s$ since if $\exists s \in S_s \cup T_s$ such that $|s|<r$ we can easily find a sets of strings $S_s',T_s'$ by extending each string of $S_s \cup T_s$ with a new characters, such that for every $s' \in S_s'\cup T_s'$ $|s'| = r$ and the overlap graph of $(S_s',T_s')$ is isomorphic to $G$. We have to find an assignment to the variables such that each constraint is satisfied. As the definition of variables $x_{u,v,i,j}$ suggests, let $x_{u,v,i,j}$ be equal to $1$ if $u(i) = v(j)$, and $0$ otherwise. For each edge $e = uv \in E(G)$ and for each $i \in \{1,\dots, r\}$ let $z_{e,i} = 1$ if $\sum_{k=i}^{r}{x_{u,v,k,k-i+1}} = r-i+1$ and $0$ otherwise. Now, let $uv \in E(G)$. Since $G$ is the overlap graph of $(S_s,T_s)$, there exists some $i \in \{1,2,\dots, r\}$ such that $\suff (u,i) = \pref (v,i)$, i.e., $\forall k \in \{i, i+1, ..., r\}$, $x_{u,v,k,k-i+1} = 1 \implies \sum_{k=i}^{r}{x_{u,v,k,k-i+1}} = r-i+1 \implies z_{e,i} = 1 \implies z_{e,1} + z_{e,2} + ... + z_{e,r} \geq 1$ is satisfied. We also need to prove that the {\textit edge constraints} of the second and third type are satisfied. If $0$ is assigned to $z_{e,i} $ for some $i \in \{1,\dots,r\}$ then $(r-i+1) \cdot z_{e,i} \leq \sum_{k=i}^{r}{x_{u,v,k,k-i+1}}$ since all variables are non-negative. If $1$ is assigned to $z_{e,i}$ for some $i \in \{1,\dots,r\}$ then $\forall k \in \{i,\dots r\}$ $x_{u,v,k,k-i+1} = 1 \implies \sum_{k=i}^{r}{x_{u,v,k,k-i+1}} = r-i+1 \geq (r-i+1)\cdot z_{e,i}$.\textit{Edge constraints} of the third type are also satisfied: if $z_{e,i} = 0$ for some $i \in \{1,\dots, r\}$ then $\exists k \in \{i, \dots, r\}$ s.t. $x_{u,v,k,k-i+1} = 0 \implies \sum_{k=i}^{r}{x_{u,v,k,k-i+1}} \leq r-i$.
If $z_{e,i} = 1$ for some $i \in \{1,\dots,r\}$ then $\sum_{k=i}^{r}{x_{u,v,k,k-i+1}} = r-i+1 = r-i + z_{e,i}$.

For each $u \in S$ and $v \in T$ such that $uv \not \in E(G)$ and for each $i \in \{1,\dots, r\}$ among variables $x_{u,v,k,k-i+1}$ for $k \in \{i,i+1, ..., r\}$ there exists at least one variable $x_{u,v,l,l-i+1}$ for some $l \in \{i, i+1, ..., r\}$ such that $x_{u,v,l,l-i+1} = 0$, otherwise, if all of them are equal to $1$, then by the definition of the $x$-variables, there will be an overlap between vertices $u$ and $v$ which would be contradicting the fact that $G$ is overlap graph of $(S_s, T_s)$.
Then $\sum_{k=i}^{r}{(1-x_{u,v,k,k-i+1})} \geq 1$ is satisfied since $1 - x_{u,v,l,l-i+1} = 1$ and $1-x_{u,v,k,k-i+1} \geq 0$ for all $k \in \{i,\dots, r\}$.
Since this is true for an arbitrary $i$, all non-edge constraints are satisfied.

Transitivity constraints are obviously satisfied: if $x_{u,v,i,j} = 1$ and $x_{w,v,k,j} = 1$ and $x_{w,q,k,l}=1$ then $u(i) = v(j) = w(k) = q(l)$ which means that also $x_{u,q,i,l} = 1$. Thus $ x_{u,v,i,j} + x_{w,v,j,k} + x_{w,q,k,l} - x_{u,q,i,l} \leq 2$ is satisfied.
In all other cases, when at least one among variables $x_{u,v,i,j}, x_{w,v,j,k}, x_{w,q,k,l}$ is equal to $0$, the sum $x_{u,v,i,j} + x_{w,v,j,k} + x_{w,q,k,l}$ is at most $2$ and since $x_{u,q,i,l}$ is non-negative variable we conclude $ x_{u,v,i,j} + x_{w,v,j,k} + x_{w,q,k,l} - x_{u,q,i,l} \leq 2$.

($\implies$) Suppose now that our integer linear program has a feasible solution. We want to prove that $r(G) \leq r$.
Assign a string of length $r$ consisting of null characters i.e., character $*$, to each vertex $u \in V(G)$ and suppose that vertices of $S$ and $T$ are linearly ordered so that we have $S = \{u_1, \dots, u_n\}$ and $T = \{v_1, \dots, v_n\}$. Then iterate over $u \in \{u_1, u_2, \dots, u_n\}$ and $v \in \{v_1, v_2, \dots, v_n\}$ and over all indices $i,j \in \{1,\dots,r\}$ and do the following:
\begin{itemize}
\item if $x_{u,v,i,j} = 1$ and $u(i) = v(j) = *$, take a character $c$ not yet used and assign $c$ as the $i$th character of $u$ and $j$th character of $v$
\item if $x_{u,v,i,j} = 1$ and $u(i) \neq *$, $v(j) = *$ then assign $u(i)$ as the $j$th character of $v$
\item if $x_{u,v,i,j} = 1$ and $v(j) \neq *$, $u(i) = *$ then assign $v(j)$ as the $i$th character of $u$.
\end{itemize}
We have to prove that the overlap graph of the set of strings obtained this way is isomorphic to $G$.
First, we will prove that such assignment is well defined.
To do this, it is enough to prove that $u(i) \neq v(j)$ and $u(i),v(j) \neq *$ is impossible, i.e., that the $i$th character of $u$ and the $j$th character of $v$ are not both different from each other and from the null character at same time.
We want to prove that such a situation leads to a contradiction.
Suppose $x_{u,v,i,j} = 1$ for some tuple $u,v,i,j$ where $u \in S$, $v \in T$ and suppose $u(i) \neq v(j) \neq *$ and suppose that this occurs for the first time i.e. $u,v$ are minimal.
Then, there exist $w \in S, w < u$, $k \in \{1,\dots r\}$ and $q \in T, q < v$, $l \in \{1,\dots,r\}$ such that $x_{u,q,i,l} = 1$ and $x_{w,v,k,j} = 1$.
Then we have $x_{u,v,i,j} + x_{u,q,i,l} + x_{w,v,k,j} - x_{w,q,k,l} \leq 2 \implies x_{w,q,k,l} = 1$ since all transitivity constraints are satisfied.
This means that it has already happened that two characters are different and the variable denoting their equality is equal to $1$.
This is the contradiction with assumption that $u,v$ are minimal.
Now, choose any edge $e = uv \in E(G)$.
We want to show that the string assigned to $u$ overlaps the string assigned to $v$.
Since $z_{e,1} + z_{e,2} + ... + z_{e,r} \geq 1$ is satisfied, there exists and index $i \in \{1,2, ..., r\}$ such that $z_{e,i} = 1$.
Then, $(r-i+1)\cdot z_{e,i} = r-i+1 \implies x_{u,v,k,k-i+1} = 1 \quad \forall k \in \{i, i+1, ..., r\}$ otherwise edge constraints of the second type would not be satisfied. This means that $u(i) = v(1)$, $u(i+1)=v(2)$, ..., $u(r) = v(i)$.
Therefore, $u$ overlaps $v$. Suppose that for some vertices $u \in S$ and $v \in T$ $uv \not \in E(G)$. Then, we have to prove that $u$ does not overlap $v$.
Since all non-edge constraints are satisfied, for each $i \in \{1,\dots, r\}$ there exists some $k \in \{i, i+1,\dots, r\}$ such that $x_{u,v,k,k-i+1} = 0$.
This means that in our procedure $u(k)$ and $v(k-i+1)$ will be different since $x_{u,v,k,k-i+1} = 0$ which means that there is no overlap of size $r-i+1$ between $u$ and $v$.
The above arguments show that $G$ is overlap graph of the constructed set $C$ which completes the proof.
\end{proof}
\section{Size of the ILP}
Let us compute the number of variables and constraints of the derived ILP for a given balanced bipartite graph $G$ with $n$ vertices on both sides and $m$ edges.
We defined variables $x_{u,v,i,j}$ for each pair of vertices $u \in S$ and $v \in T$ and for each pair of indices over $\{1,\dots,r\}$.
Thus, we have $n\cdot n \cdot r \cdot r = n^2r^2$ $x$-variables.
Also, for each edge $e$ and for each index from the set $\{1,\dots,r\}$ we defined a variable $z_{e,i}$ which gives us $|E(G)|\cdot r$ $z$-variables.
In total, we have $n^2r^2 + r|E(G)|$ variables.

For each edge, we defined one constraint consisting of variables $z$ and $2r$ constraints of the second type.
For each non-edge we defined $r$ constraints.
We defined one constraint for each $u,w \in S$, $v,q \in T$, with $u\neq w$ and $v \neq q$, and for each $i,j,k,l \in \{1, \dots, r \}$ which gives us ${n\choose 2} \cdot {n\choose 2} \cdot r^4$.
In total, we have $(1 + 2r) m + r(n^2-m) + {n\choose 2}{n\choose 2}r^4 = m + r n^2 + r m + {n\choose 2}{n\choose 2}r^4$ constraints.

Observe that the number of transitivity constraints can be huge for graphs with large number of vertices. So, unfortunately, in practice this model can be used only for small graphs. For example consider a bipartite graph $G=(V,E)$ with $|V(G)| = n = 10$, $|E(G)| = m = 50$ and $r = 10$. The number of variables is $1050$, the number of constraints is $20251550$.

\chapter{An integer programming formulation for digraphs}

In the bipartite model, we assumed without loss of generality that if a given graph $G$ is of readability at most $r$ then we can find a sets $S_s$ and $T_s$ such that each string in $S_s \cup T_s$ is of the length exactly $r$.
In the model for digraphs, such assumption would not be without loss of generality. An example for this is given in Figure~\ref{fig:digraph_cnt_same_length}.

\begin{figure}[h!]
\begin{center}
\includegraphics[width=0.4\linewidth]{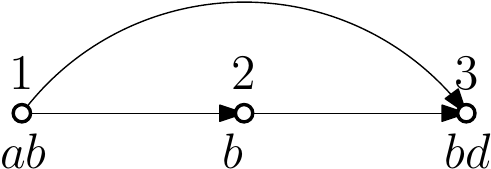}
\end{center}
\caption{Digraph of readability 2 for which there is no overlap labeling such that all strings from the image of labeling are of length 2.}\label{fig:digraph_cnt_same_length}
\end{figure}

\section{Variables of the ILP}

Let a digraph $D = (V,A)$ and a positive integer $r$ be given. Consider the same decision problem as above: is $D$ of readability at most $r$?
To be able to answer this question, we introduce an integer linear program.
The idea is similar to the integer linear program defined for balanced bipartite graphs except that we need new variables modeling how long the individual strings are and that we need to assure injectivity of the labeling.
For every two distinct vertices $u,v \in V(D)$ and for every two indices $i,j \in \{1,\dots,r\}$ we introduce a binary variable $x_{u,v,i,j}$ which is equal to $1$ if the $i$th character of the string assigned to vertex $u$ is equal to the $j$th character of the string assigned to $v$ and not equal to the null character, i.e., $u(i) = v(j)\neq *$, and $0$ otherwise.
Also, for every vertex $u \in V(D)$ and for every index $i \in \{1,\dots,r\}$ we introduce a binary variable $t_{u,i}$ which is equal to $1$ if $\forall k \in \{i,\dots,r\}$ $u(k) = *$ and zero otherwise. Recall that $*$ stands for the null character.
Variable $t_{u,i}$ is equal to $1$ if and only if the characters at positions $i,i+1,\dots,r$ of the string assigned to vertex $u$ are equal to the null character.
The third type of the variables for the ILP are $z$-variables. For every edge $e = uv \in A(D)$, for every index $i \in \{1,\dots,r\}$, and for every index $l \in \{i,\dots,r\}$, we introduce a binary variable $z_{e,i,l}$.
The variable $z_{e,i,l}$ equals to $1$ if $\forall k \in \{i,\dots,l\}$ $x_{u,v,k,k-i+1} = 1$ and $0$ otherwise.

\section{Constraints of the ILP}

\begin{itemize}
\item \textbf{Edge constraints:} for each edge $e = uv \in A(D)$ add the following constraints:
	\begin{enumerate}
		\item \begin{equation}\label{const:e_fst_type}
		\sum_{i=1}^{r}{\sum_{l=i}^{r}{z_{e,i,l}}} \geq 1
		\end{equation}
		\item $\forall i \in \{1,\dots,r\}$, $\forall l \in \{i,\dots,r\}$:
			\begin{equation}\label{const:e_snd_type}
			\sum_{k=i}^{l}{x_{u,v,k,k-i+1}} \geq (l-i+1)\cdot z_{e,i,l}
			\end{equation} and
			\begin{equation}\label{const:e_thrd_type}
			t_{u,l+1} - z_{e,i,l} \geq 0
			\end{equation}
	\end{enumerate}
\item \textbf{Non-edge constraints:} for every two vertices $u,v \in V(D)$ such that $uv \not \in A(D)$ add the following constraints:
	\begin{enumerate}
		\item $\forall i \in \{1,\dots,r-1\}$, $\forall l \in \{1,\dots,i\}$
		\begin{equation}\label{const:ne_fst_type}
		t_{u,i+1} + \sum_{k=l}^{i}{x_{u,v,k,k-l+1}} \leq i-l+1
		\end{equation}
		\item $\forall l \in \{1,\dots,r\}$
		\begin{equation}\label{const:ne_snd_type}
		\sum_{k=l}^{r}{x_{u,v,k,k-l+1}} \leq r-l
		\end{equation}
	\end{enumerate}
\item \textbf{Transitivity constraints: }
	for all pairwise distinct $u,v,w \in V(D)$ and for all indices $i,j,k \in \{1,\dots,r\}$ add the following constraints:
		\begin{equation}\label{const:a_fst_type}
		x_{u,v,i,j} + x_{v,w,j,k} - x_{u,w,i,k} \leq 1
		\end{equation}
\item \textbf{$t$-monotonicity constraints: }
		for all vertices $u \in V(D)$ and for all $i \in \{1,\dots,r-1\}$ add the following:
		\begin{equation}\label{const:a_snd_type}
		t_{u,i} \leq t_{u,i+1}
		\end{equation}
		
\item \textbf{Injectivity constraints: }
	for all distinct vertices $u,v \in V(D)$:
		\begin{equation}
		\sum_{i=1}^r{x_{u,v,i,i}} \leq r-1	
		\end{equation}
\item \textbf{Symmetry constraints: } for all distinct vertices $u,v \in V(D)$ and for all indices $i,j \in \{1,\dots,r\}$:
		\begin{equation}
			x_{u,v,i,j} = x_{v,u,j,i}
		\end{equation}
\item \textbf{Additional constraints: }
for all vertices $u,v \in V(D)$ and for all indices $i,j \in \{1,2,\dots,r\}$:
		\begin{equation}\label{const:a_thrd_type}
		1-x_{u,v,i,j} \geq t_{v,j}
		\end{equation}
		\begin{equation}\label{const:a_fourth_type}
		1-x_{u,v,i,j} \geq t_{u,i}
		\end{equation}
\end{itemize}
\section{Objective function}
We look for a feasible solution of above constraints, so we can define an objective function to be any linear function, for example the constant zero function.

\section{Description of the ILP}

\begin{theorem}
For a given digraph $D = (V,A)$ and a positive integer $r$, the above defined integer program has a feasible solution if and only if $D$ is of readability at most~$r$.
\end{theorem}
We omit a formal proof but describe the main ideas of the proof. The ILP formulation for the exact computation of readability for digraphs is (as already stated) very similar to the ILP for balanced bipartite graphs. The only difference is that we have to ensure injectivity and allow strings of different lengths. For that reasons, we introduced new set of variables, $t$-variables. \emph{Edge constraints} ensure that there is an overlap between any two adjacent vertices. The variable $z_{e,i,l}$ for some edge $e \in A(D)$ and indices $i,l \in \{1,\dots,r\}$ (here we have in mind the same digraph $D$ and integer $r$ as in the formulation of the decision problem introduced at the beginning of this chapter) intuitively means that characters at positions $i,\dots,l$ of a string assigned to vertex $u$ are equal to characters at positions $1,2,\dots,i-l+1$ of a string assigned to vertex $v$ (respectively) and not equal to the null character. Since $e \in A(D)$ we want that at least one variable among $z_{e,i,l}$ for $i\in \{1,\dots,r\}$ and $l \in \{i,\dots,r\}$ is equal to $1$. This is ensured by summing all $z$-variables for a fixed edge and setting that it must be at least 1.
The second type of edge-constraints (inequality \eqref{const:e_snd_type}) ensures us that $z_{e,i,l}$ can not be equal to $1$ if not all characters at positions $i,\dots,l$ are equal to characters at positions $1,\dots,i-l+1$ of a strings $s_u$ and $s_v$.
Inequality \eqref{const:e_thrd_type} ensures us that $z_{e,i,l}$ can be equal to $1$ if all characters at positions $l+1,\dots,r$ are equal to null character.

The \textit{non-edge constraints} we introduce to forbid an overlapping between non-adjacent vertices. Let $i \in \{1,\dots,r-1\}$ and $l \in \{1,2,\dots,i\}$ be fixed.
We have two cases.
The trivial case is if $t_{u,i+1} = 0$.
Then there is at least one character of a string $s_u$ different from the null character at positions starting at $i+1$ (see also \textit{additional constraints}).
So there is no overlap of length $i-l+1$ (for fixed $i,l$).
The second, more difficult case is if all characters at positions $i+1,\dots,r$ are equal to the null character.
In that case $t_{u,i+1}=1$ and since we do not have an overlapping of a strings $s_u$ and $s_v$, in order to have inequality~\eqref{const:ne_fst_type} satisfied, at least one of the variables $x_{u,v,k,k-l+1}$ for $k \in \{l,\dots,i\}$ must be equal to $0$, because otherwise the sum on the left side of \eqref{const:ne_fst_type} will be equal to $i-l+1$ and since $t_{u,i+1}=1$ we have the left side equal to $i-l+2$ and the right side equal to $i-l+1$.
Note that in inequality~\eqref{const:ne_fst_type} index $i \in \{1,\dots,r-1\}$ (it can not take value $r$).
For that, boundary case, we introduce inequlaity~\eqref{const:ne_snd_type}, the idea is the same as above.

The role of other constraints can be explained as follows.
\begin{enumerate}
\item (\textbf{Transitivity constraints}) Let $u,v,w \in V(D)$ be some distinct vertices and let $i,j,k \in \{1,\dots,e\}$. The inequality~\eqref{const:a_fst_type} is introduced to ensure the following implication: if $s_u(i) = s_v(j)$ and $s_v(j) = s_w(k)$ then $s_u(i) = s_w(k)$.
\item (\textbf{$t$-monotonicity constraints}) Let $u \in V(D)$ and $i \in \{1,\dots,r-1\}$. If all characters at positions $i,\dots,r$ of a string $s_u$ are equal to the null character i.e., $t_{u,i} = 1$ then also all characters at positions $i+1,\dots,r$ are equal to the null character. In other words $t_{u,i} = 1 \implies t_{u,i+1} = 1$. This is ensured with inequality~\eqref{const:a_snd_type}.
\item The role of \textbf{injectivity} and \textbf{symmetry} constraints is obvious.
\item (\textbf{Additional constraints}) The inequalities~\eqref{const:a_thrd_type} and \eqref{const:a_fourth_type} ensure that we do not have overlapping of some length such that it consists of null characters.

\end{enumerate}

\chapter{Bounds on readability}

In Chapter 2 we have presented an upper bound for readability for an arbitrary directed graph $G$ is given by $2^{p+1} - 1$ where $p=\max\{\Delta^+(G), \Delta^-(G)\}$.
This bound can be very weak. For example, for the complete bipartite graph $G$ with $n$ we get $r(G) \leq 2^{n+1} - 1$ while the readability of $G$ is $1$. An overlap labeling of length $1$ is obtained by assigning the same character to each vertex.
In this chapter we present lower and upper bounds for the readability of bipartite graphs using characterizations given below as well as characterization of the graphs of readability at most 2.

\section{Graphs of readability at most 2}
In this section we will present graph theoretic characterizations of the graphs with readability 1 or 2 following Kratsch et al.~\cite{MilanicNotPublished}.
\begin{theorem}\label{thm:read1}
Let $G = (S\cup T, E)$ be a bipartite graph. The following is equivalent:
\begin{enumerate}
\item $r(G) = 1$.
\item $G$ is a disjoint union of bicliques.
\end{enumerate}
\end{theorem}

\begin{proof}
(1) $\implies$ (2): Suppose that $r(G) = 1$. Let $\ell$ be an overlap labeling of $G$ of the size $1$ and let $k$ be the number of different characters used by labeling $\ell$. Denote the characters by $a_1, \dots, a_k$. Construct the following graphs: $\forall i \in \{1, \dots, k\}$ let $G_i$ be the subgraph of $G$ induced by $S_i \cup T_i$ where $S_i = \{u \in S \mid \ell(u) = a_i \}$ and $T_i = \{v \in T \mid \ell(v) = a_i \}$. $G_i$ is obviously biclique because there is an overlap in $G$ between any two vertices $u \in S_i$ and $v \in T_i$ so it must hold that there is an edge between any two of them. From the construction of graphs $G_i$ it is clear that any vertex $u \in S$ is contained in exactly one $S_i$ for some $i \in \{1,\dots,k\}$ since $\ell(u) = a_i$. The same argument we can apply for vertices in $T$. Combining the above, we get $G_1 + \dots + G_k = G$.

(1) $\impliedby$ (2): Suppose $G$ is disjoint union of graphs $G_1=(S_1\cup T_1,E_1), \dots, G_k=(S_k\cup T_k, E_k)$ with $S = \cup_{i=1}^{k}S_i$ and $T = \cup_{i=1}^kT_i$ and let $a_1, \dots a_k$ be pairwise distinct characters. Then $\forall i \in \{1,\dots,k\}$ label all vertices of $G_i$ with the same character $a_i$. That way we made an overlap between any two vertices that belong to $G_i$. Since $G$ is the disjoint union of $G_i$s there are no edges between $G_i$ and $G_j$ in $G$ for $i \neq j$ and also by the above construction there is no overlap since $a_i \neq a_j$ for $i \neq j$. Thus, such assignment is an overlap labeling of length $1$, which completes the proof.
\end{proof}

Using the above characterization one can find a polynomial time algorithm to recognize whether a bipartite graph has readability $1$. An example is given in Figure~\ref{fig:alg_exec1}.

\begin{figure}[h!]
\begin{center}
\includegraphics[width=0.3\linewidth]{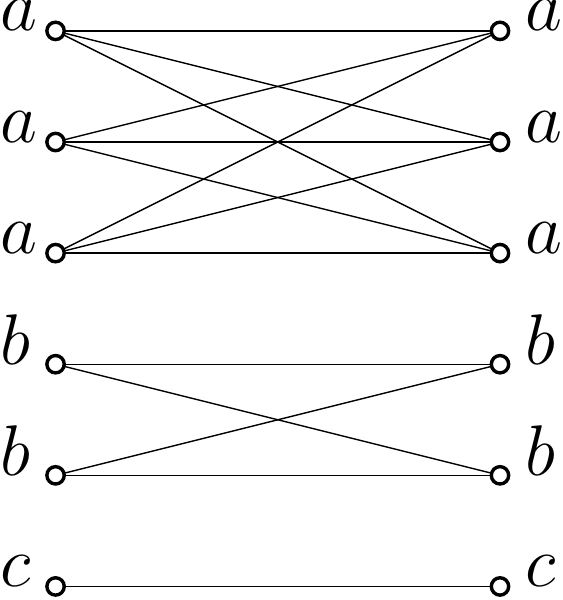}
\end{center}
\caption{Example of a bipartite graph of readability 1.}\label{fig:read_ex_1}
\end{figure}

In order to state the characterization of bipartite graphs with readability at most 2 we need to introduce the following definitions.
\begin{definition}
A bipartite graph $G$ is called \emph{twin-free} if the following implication holds: $(\forall \, u,v \in V(G))$ $(N(u) = N(v) \implies u = v)$.
\end{definition}

\begin{definition}\label{def:feasible_matching}
A matching $M$ in a bipartite graph $G$ is said to be \emph{feasible} if and only if the following conditions are satisfied:
\begin{enumerate}
\item The graph $G-M$ is a disjoint union of bicliques.
\item For every induced subgraph $F$ of $G$ isomorphic to $C_6$ (cycle on 6 vertices), we have $|M \cap E(F)| = 3$. In other words, if the edges of $F$ are labeled as in the Figure~\ref{fig:c6_domino}, then $M \cap E(F) \in \{\{e_1,e_3,e_5\}, \{e_2,e_4,e_6\}\}$.
\item For every induced subgraph $F$ of $G$ isomorphic to the domino (see Definition~\ref{def:domino}) with edges labeled as in Figure~\ref{fig:c6_domino} we have $M\cap E(F) \in \{ \{e_2,e_6\}, \{ e_3, e_5 \} \}$.
\end{enumerate}
\end{definition}

\begin{figure}[h!]
\begin{center}
\includegraphics[width=0.35\linewidth]{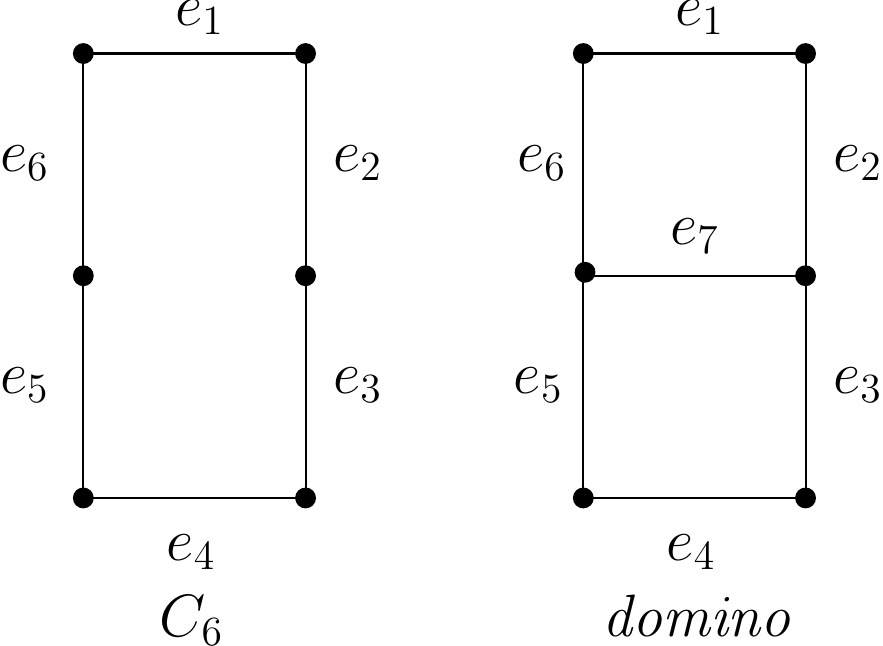}
\end{center}
\caption{Labeling of edges of the $C_6$ and the domino used in Definition~\ref{def:feasible_matching}.}\label{fig:c6_domino}
\end{figure}

\begin{theorem}[Kratsch et al.~\cite{MilanicNotPublished}]\label{thm:chrrdb2}
Let $G$ be a twin-free bipartite graph. Then, $r(G) \leq 2$ if and only if $G$ has a feasible matching.
\end{theorem}

The condition that $G$ must be twin-free graph is not necessary. The proof for that is a simple corollary of Lemma~\ref{lemma:twin_reduction}.

\begin{lemma}\label{lemma:twin_reduction}
Let $G$ be a bipartite graph with two distinct vertices $u$ and $v$ such that $N(u) = N(v)$. Then, $r(G) = r(G-u)$.
\end{lemma}

Thus, if $G$ is not twin-free, then we can define its twin-free reduction $TF(G)$ which consists of equivalence classes of $G$ with respect to twin relation $u \sim v \iff N(u) = N(v)$ and classes $U$ and $V$ are adjacent if and only if $uv \in E(G)$ for some $u \in U$ and $v \in V$. With that reduction we obtain a twin-free bipartite graph $H$ and by the above lemma we have $r(H) = r(G)$ so we can use characterization given by Theorem~\ref{thm:chrrdb2} to check if $G$ is of readability at most two. A polynomial time algorithm reducing the problem to \textit{2-SAT} is presented in \cite{MilanicNotPublished}.

\section{Lower and upper bounds on readability}

In order to give lower bound for the readability of bipartite graphs, we introduce a graph parameter \textit{distinctness}.
\begin{definition}
Let $G = (V,E)$ be an arbitrary graph. The \textit{distinctness} of $u,v \in V(G)$ denoted by $\dt(u,v)$, is defined as $\dt(u,v) = \max\{ |N(u)\backslash N(v)|, |N(v)\backslash N(u)| \}$. The \textit{distinctness of a bipartite graph $G=(S \cup T, E)$} is denoted by $\dt(G)$ and given by $\dt(G) = \min_{u,v \in S, u,v \in T}\{ \dt(u,v) \}$, that is, the minimum distinctness of any pair of vertices that belong to the same part of the bipartition.
\end{definition}
Note that the distinctness of a given bipartite graph can be computed in polynomial time.

\begin{theorem}[Chikhi at al.~\cite{MilanicRayanMedvedev}]\label{thm:lowerb}
For an arbitrary bipartite graph $G$ of maximum degree at least two, $r(G) \geq \dt(G)+1$.
\end{theorem}

Using the above theorem, it is shown in~\cite{MilanicRayanMedvedev} that there exist graphs of readability at least linear in the number of vertices.
The bipartite graph $G = (S\cup T, E)$ with $2^n-1$ vertices on both sides with this property is obtained by the following rules:
\begin{enumerate}
\item $S = \{v_s \mid v \in \{0,1\}^n \backslash \{0\}^n\}$ and $T = \{v_t \mid v \in \{0,1\}^n \backslash \{0\}^n\}$
\item $
		E(G) = \left\{ (u_s, v_t) \in S \times T \mid \sum_{i=1}^n{u[i]\cdot v[i]} \equiv 1 \pmod 2 \right\}$
\end{enumerate}
One can show that for the graph $G$ obtained this way we have $r(G) \geq \frac{2\cdot(2^n - 1)}4 = \frac{|V(G)|}4$. The idea of the proof is to show first that if two distinct vertices in the same part of the bipartition of $G$ have a common neighbor, then they have exactly $\frac{|V(G)|}{4}$ neighbors. Combining this argument with Theorem~\ref{thm:lowerb} we get the desired conclusion. These graphs are interesting since for any positive integer $r$ one can construct a bipartite graph of readability at least $r$ (although, this can be done more easily also with trees \cite{MilanicRayanMedvedev}).

\begin{sloppypar}
We now introduce the notions of \textit{decomposition} of a bipartite graph and the \textit{hierarchial-union-of-bicliques} rule (shortly \textit{HUB-rule}) for a decomposition (and consequently the \textit{HUB-number} of a bipartite graph) in order to be able to improve upper bound given by Braga and Meidanis~\cite{BragaMeidanis}. The rule is introduced in the paper \cite{MilanicRayanMedvedev}. Also, we will use these notions to explore the readability of some special graphs.
\end{sloppypar}
\begin{definition}
Let $G$ be a graph. The \textit{decomposition} of size $k$ of  $G$ is a function $w : E(G) \to \{1,\dots,k\}$.
\end{definition}

Intuitively, the decomposition of $G$ is introduced in order to split the set of edges of a graph into pairwise disjoint sets according to the minimum length of the strings assigned to the endpoints of edges. Given a bipartite graph $G = (S \cup T, E)$ and a decomposition $w$ of $G$ of size $k$, a partition of the set of edges to pairwise disjoint sets $E_1, \dots, E_k$ is defined by $E_i = \{e \in E(G) \mid \omega(e) = i\}$ for all $i \in \{1,\dots,k\}$. Consequently, one can define graphs $G_1^w = (S \cup T, E_1^w), \dots, G_k^w = (S \cup T, E_k^w)$ with respect to $w$ where $E_k^w = \{ e \in E(G) \mid w(e) = i \}$.

\begin{figure}[h!]
\begin{center}
\includegraphics[width=1\linewidth]{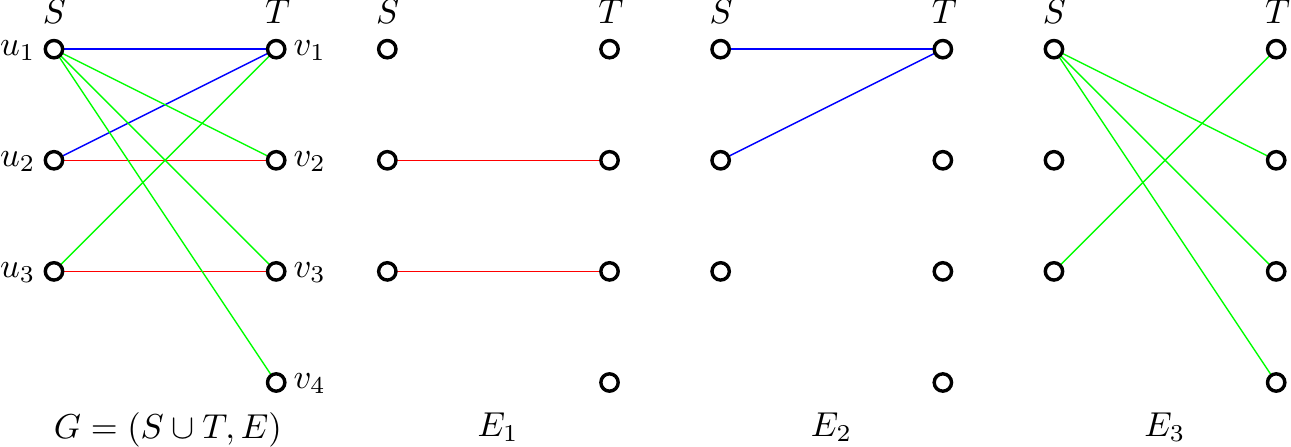}
\end{center}
\caption{Example of a decomposition and the resulting partitioning the set of edges of a bipartite graph $G$}\label{fig:decomposition}
\end{figure}

Figure~\ref{fig:decomposition} provides an example of a decomposition of size $3$. $E_1$ consists of red edges, mapped to $1$, $E_2$ consists of blue edges, mapped to $2$ and $E_3$ consists of green edges, mapped to $3$. The following holds: $E = E_1\cup E_2\cup E_3$.

We say that two distinct vertices of a graph are \textit{twins} if they have the same (open) neighborhoods. A vertex of a graph is \textit{non-isolated} if its (open) neighborhood is non-empty.
\begin{definition}\label{def:hub_rule}
Let $G = (S\cup T, E)$ be a bipartite graph and let $w$ be a decomposition of size $k$. We say that $w$ satisfies the \textit{HUB-rule} if the following hold:
\begin{enumerate}
\item for all $i \in \{1,\dots,k\}$, $G_i^{w}$ is a disjoint union of bicliques and
\item if two distinct vertices $u \in S$ and $v \in T$ are non-isolated twins in $G_i^{w}$ for some $i \in \{2,\dots,k\}$ then for all $j \in \{1,\dots,i-1\}$ $u$ and $v$ are twins in $G_{j}^{w}$.
\end{enumerate}
\end{definition}

In what follows, if a decomposition $w$ is clear from the context, we will usually write $G_i$ instead of $G_i^w$.
An example of a decomposition of a bipartite graph $G$ that satisfies the HUB-rule is any $w : E(G) \to \{1,\dots,k\}$ such that $G_1^w$, defined as above, is a disjoint union of bicliques and $E_2^w, \dots E_k^w$ are matchings in $G$. This is true since by definition $G_1^w$ is a disjoint union of bicliques and since $E_2^w, \dots, E_k^w$ are matchings in $G$ then there are no pairs of twin vertices. This example is the main motivation for exploring the readability of grid graphs introduced in Chapter 7.

Consider now a bipartite graph $G = (S \cup T, E)$ and an overlap labeling $\ell$ of $G$. Then one can define decomposition $w$ such that for any edge $e = uv \in E(G)$, $w(e) = ov(\ell(u),\ell(v))$. Clearly, $w$ is well defined. The decomposition obtained this way will be called \textit{$\ell$-decomposition}. The next theorem gives a connection between $\ell$-decompositions of $G$ and the HUB-rule:
\begin{theorem}\label{thm:ldcmhubrule}
Let $\ell$ be an overlap labeling of a bipartite graph $G = (S\cup T, E)$. Then, the $\ell$-decomposition satisfies the HUB-rule.
\end{theorem}

A proof of Theorem~\ref{thm:ldcmhubrule} can be found in~\cite{MilanicRayanMedvedev}.
It is easy to see that the $\ell$-decomposition given in Figure~\ref{fig:decom_hub_rule} satisfies the HUB-rule (in fact, it satisfies the conditions stated in the paragraph following Definition~\ref{def:hub_rule}). Figure~\ref{fig:decomposition} provides an example of a decomposition that does not satisfy the HUB-rule: clearly, each of the graphs $G_1, G_2, G_3$ is a disjoint union of bicliques but, vertices $u_1$ and $u_2$ are twins in $G_2$ and not in $G_1$ and thus the second condition of Definition~\ref{def:hub_rule} is not met.

\bigskip
\bigskip
\begin{figure}[h!]
\begin{center}
\includegraphics[width=1\linewidth]{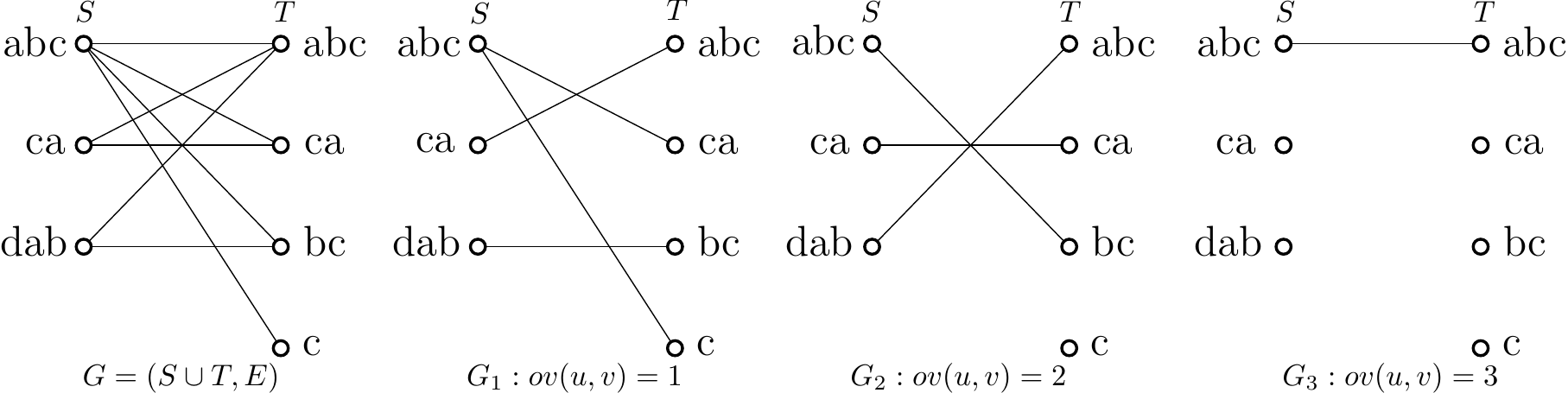}
\end{center}
\caption{Example of $\ell$-decomposition for an overlap labeling $\ell$ of a graph $G$. Each of $G_1,G_2,G_3$ is a  disjoint union of bicliques.} \label{fig:decom_hub_rule}
\end{figure}

Motivated by the HUB-rule, Chikhi at al.~\cite{MilanicRayanMedvedev} introduced the HUB-number of a bipartite graph $G$.
\begin{definition}
The \emph{HUB-number} of a bipartite graph $G$, denoted by $\hub(G)$, is the minimum positive integer $k$ such that there exists a decomposition of size $k$ which satisfies the HUB-rule.
\end{definition}

One of the motivations for introducing the HUB-number of a bipartite graph is the following corollary of Theorem~\ref{thm:ldcmhubrule}:
\begin{corollary}
Let $G$ be a bipartite graph. Then $r(G) \geq \hub(G)$.
\end{corollary}

\begin{proof}
Let $G$ be a bipartite graph of readability $r$. Then there is an overlap labeling $\ell$ of $G$ of size $r$. The $\ell$-decomposition is of size $r$ and satisfies the HUB-rule by Theorem~\ref{thm:ldcmhubrule} which implies $r(G) = r \geq \hub(G)$.
\end{proof}
An upper bound for the readability by means of the HUB-number is captured by the next theorem:
\begin{theorem}[Chikhi at al.~\cite{MilanicRayanMedvedev}]\label{thm:upbrd}
Let $G$ be a bipartite graph. Then $r(G) \leq 2^{\hub(G)} -1$.
\end{theorem}

We now present lemma which will show that the upper bound given in Theorem~\ref{thm:upbrd} is at most equal to the upper bound $r(G) \leq 2^{\Delta(G)+1}-1$ from \cite{BragaMeidanis}.

\begin{lemma}\label{lemma:hubrg}
Let $G = (S \cup T, E)$ be a bipartite graph. Then $\hub(G) \leq \Delta(G)$.
\end{lemma}

\begin{proof}
Let $G = (S \cup T, E)$ be a bipartite graph. Let $L = \{G_1, G_2, \dots G_k\}$ for some positive integer $k$ be a decomposition of $G$ into matchings, i.e., for all $i \in \{1,\dots,k\}$ $G_i$ has the same set of vertices as $G$ and the set of edges of $G_i$ is matching in $G$. The decomposition $w$ defined as $G_i^w = G_i$ for $i \in \{1,\dots,k\}$ satisfies HUB-rule since: condition (i) of definition of HUB-rule is obviously satisfied since $G_i^\omega$ for all $i \in \{1,\dots,k\}$ are matchings of $G$ when isolated points of $G_i$ are removed. Condition (ii) is satisfied since for any $i \in \{1,\dots,k\}$ there is no vertices which are twins. By Theorem~\ref{thm:Konig}, $G$ can be decomposed into $\Delta(G)$ matchings, which, by the above arguments, gives the decomposition of size at most $\Delta(G)$ that satisfies the HUB-rule.
\end{proof}

Clearly, combining Theorem~\ref{thm:upbrd} and Lemma~\ref{lemma:hubrg} we can see that the upper bound $2^{\Delta(G)+1}-1$ is improved. As an example, consider complete bipartite graph $K_{n,n}$ with $n$ vertices on both sides. The upper bound for the readability of $K_{n,n}$ given by~\cite{BragaMeidanis} is equal to $2^{n+1}-1$ while the upper bound given by Theorem~\ref{thm:upbrd} is equal to $1$, since $\hub(K_{n,n}) = 1$ achieved by the decomposition mapping each edge to $1$. However, the complexity of calculating the HUB-number is not known in general, which limits the use of corresponding bounds either just in theory or for some special graph classes.

\chapter{Readability of two-dimensional grid graphs}

In this chapter we will explore the readability of \textit{two-dimensional grid graphs}.  We will show that readability is at most $3$ and moreover, we present a polynomial time algorithm for constructing an optimal overlap labeling of such graphs. Sometimes we will write just \textit{grid} or \textit{grid graph} when referring to the two-dimensional grid graph.

\begin{definition}
The \emph{Cartesian product} of graphs $G$ and $H$ is the graph $G \Box H$ such that:
\begin{enumerate}
\item The vertex set of $G\Box H$ is the Cartesian product $V(G) \times V(H)$.
\item Two vertices $(u,u')$ and $(v,v')$ are adjacent in $G\Box H$ if and only if
\begin{itemize}
\item $u = v$ and $u'$ is adjacent to $v'$ in $H$, or
\item $u' = v'$ and $u$ is adjacent to $v$ in $G$.
\end{itemize}
\end{enumerate}
\end{definition}

\begin{definition}
A two-dimensional grid graph $G_{m,n}$ of size $m\times n$ is the Cartesian product $P_m \Box P_n$ of paths on $m$ and $n$ vertices.
\end{definition}

Observe that a grid graph of an arbitrary size is bipartite.
An example of a grid graph is given in Figure~\ref{fig:grid_graph_4x4}.

\bigskip
\begin{figure}[h!]
\begin{center}
\includegraphics[width=0.5\linewidth]{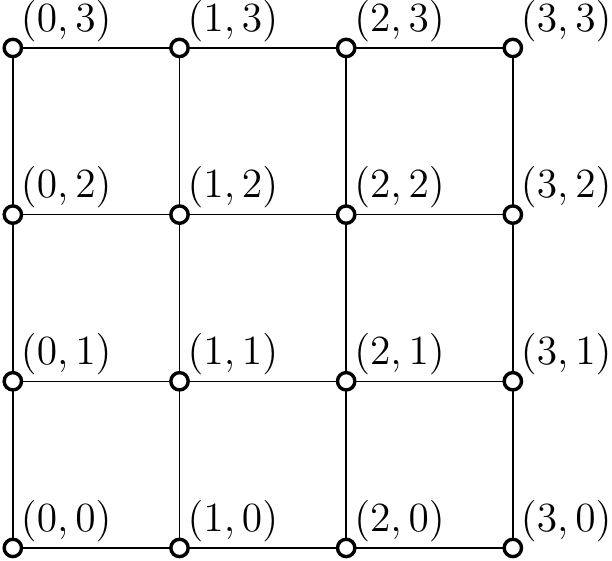}
\end{center}
\caption{An example of a two-dimensional grid graph: the $4\times 4$ grid, that is, the Cartesian product of two copies of $P_4$ paths. The vertices are denoted as described on page \pageref{lemma:lbcgn}.}\label{fig:grid_graph_4x4}
\end{figure}

\section{A polynomial time algorithm for constructing an optimal overlap labeling of grids}

\begin{theorem}\label{thm:rgrid}
Consider two integers $m \geq 3$ and $n \geq 3$ and let $G$ be the grid graph of size $m\times n$. Then $r(G) = 3$.
\end{theorem}

In order to be able to prove the above theorem, we introduce the following lemma:

\begin{lemma}\label{thm:rdbind}
Let $G$ be a graph and let $H$ be an induced subgraph of $G$. Then $r(H) \leq r(G)$.
\end{lemma}

\begin{proof}
Let $G$ be a graph of readability $r$ and let $\ell$ be an overlap labeling of $G$ of length $r$. We define a labeling $\ell_H$ of $H$ as follows: for any vertex $u \in V(H)$, $\ell_H(u) = \ell(u)$. By the definition of induced subgraph, we conclude that $\ell_H$ is an overlap labeling of $H$. Obviously, $len(\ell_H) \leq len(\ell)$ which completes the proof.
\end{proof}

The idea of the proof of Theorem~\ref{thm:rgrid} is to show that the lower bound on the readability of a grid graph is $3$ and then find an overlap labeling of length $3$. In order to prove the lower bound using Theorem~\ref{thm:rdbind} we will introduce the following graph (already mentioned in the previous chapter):
\begin{definition}\label{def:domino}
The \emph{domino} is the grid graph of size $2\times 3$.
\end{definition}
Consider the graph $F$ obtained by taking the graph $F = K_1 + domino$ and adding an edge between $K_1$ and one of the vertices of the domino with degree $3$.
Graph $F$ is also bipartite, see Figure~\ref{fig:graph_F}.

\begin{figure}[h!]
\begin{center}
\includegraphics[width=0.25\linewidth]{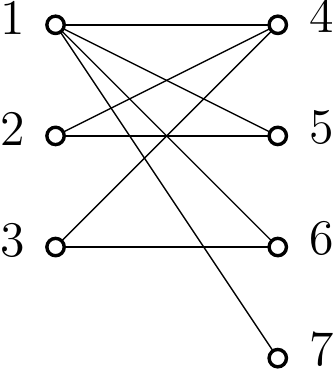}
\end{center}
\caption{An example of bipartitioning of graph $F$}\label{fig:graph_F}
\end{figure}

\begin{lemma}\label{lemma:rgd}
The readability of $F$ is 3.
\end{lemma}

\begin{proof}
We will prove the lemma by contradiction. We use the same denotation of vertices as in Figure~\ref{fig:graph_F}. Suppose that there exists an overlap labeling $\ell$ of $F$ with $len(\ell) \leq 2$.
We may assume without loss of generality that all strings assigned to the vertices of $F$ are of length 2 and moreover, we assume that $\ell(1) = ab$ for some distinct characters $a$ and $b$ since if $a = b$ then we can set $\ell(1) = ca$ for $c \neq a$ and obtain an overlap labeling of $F$ of length 2 and obtain the desired.
Observe that $\ell(1) \neq \ell(2)$ since the open neighborhoods are different. The same argument is used to prove that $\ell(1) \neq \ell(3)$ and $\ell(2) \neq \ell(3)$. Similarly, the strings $\ell(4)$ and $\ell(7)$ are pairwise distinct.
\begin{enumerate}
\item Suppose that $ov(1,4) = 2$ and $ov(1,6) = 2$. Then $\ell(4) = ab$ and $\ell(6) = ab$. Since $\ell(1) \neq \ell(3)$, the second character of vertex $3$ must be equal to $a$ in order to have an overlap between $3$ and $4$. Let $\ell(3) = ca$ for some character $c$. Since $\ell(2) \neq \ell(3)$ and we must have an overlap between $2$ and $4$, it must hold that $\ell(2) = da$ for some character $d$. But then we have an overlapping between $2$ and $6$ which is not allowed since $2$ and $6$ are not adjacent, so we obtained a contradiction.

\item Suppose that $ov(1,4) = 2$ and $ov(1,6) = 1$. Then $\ell(4) = ab$ and the first character of vertex $6$ is equal to $b$. By the above $\ell(1) \neq \ell(3)$ which imply that the second character of $3$ is equal to $a$ in order to have an overlap between $3$ and $4$. Then, the only possibility to have an overlap between $3$ and $6$ is $\ell(3) = \ell(6) = ba$. Apply similar arguments as above to obtain $\ell(2) = ca$ for some character $c \neq b$. The first character of $5$ can not be equal to $a$ since we would have $ov(3,5) > 0$. Since $2$ and $5$ are adjacent it must hold that $\ell(5) = ca$. But then, there is no overlapping of length at most 2 between vertices $1$ and $5$ (recall that we supposed $a \neq b$), a contradiction.

\item Suppose now that $ov(1,4) = 1$ and $ov(1,6) = 2$. The first character of $4$ is equal to $b$ and $\ell(6) = ab$. Since $\ell(3) \neq \ell(1)$, the second character of $3$ is equal to $a$ and in order to have an overlap between $3$ and $4$ it must hold that $\ell(3) = \ell(4) = ba$.
Since $\ell(2) \neq \ell(3)$ and $2$ and $4$ are adjacent we conclude $\ell(2) = cb$ for some $c \neq a$.
Furthermore, $1$ and $7$ are adjacent but the first character of $7$ can not be equal to $b$ since we would have $ov(2,7) > 0$ and $\ell(7) \neq ab$ because we would have $ov(3,7) > 0$. Then an overlapping between $1$ and $7$ of length at most 2 is not possible, a contradiction.

\item Suppose that $ov(1,4) = 1$ and $ov(1,6) = 1$. Then the first character of $4$ and the first character of $6$ are both equal to $b$. The second character of $2$ must be different from $b$ since otherwise we would have $ov(2,6) > 0$. Since 2 and 4 are adjacent, we infer that $\ell(2) = \ell(4) = bd$ for some character $d \neq b$. Since $\ell(3) \neq \ell(2)$ and 3 is adjacent to 4 we infer that the second character of 3 is $b$. Since 2 and 5 are adjacent and 3 and 5 are not adjacent, the only possibility is that $ov(2,5) = 1$ i.e., the first character of vertex 5 is equal to $d$. By the above inequality, $d \neq b$ and since 1 and 5 are adjacent we have $ov(1,5)=2$, which implies $\ell(5) = db$ with $d = a$. Since 1 and 7 are adjacent we must have $ov(1,7) > 0$. But if $ov(1,7)= 1$ then the first character of 7 is equal to $b$ which is not possible since it would imply an overlapping between 3 and 7 which are not adjacent. If $ov(1,7) = 2$ then $\ell(7) = ab$ and since $d=a$ we have an overlapping between 2 and 7 which is not allowed. Thus, there is no overlap between 1 and 7 of length at most 2, a contradiction.
\end{enumerate}
By the above analysis, we conclude that $r(F) \geq 3$. An overlap labeling of length 3 is given in Figure~\ref{fig:decom_hub_rule} (first image). Thus, $r(F) = 3$ as claimed.
\end{proof}

Recall that the grid graph of size $m\times n$ is defined as the Cartesian product of the paths $P_m$ and $P_n$. If we denote the vertices of $P_m$ with $\mathbb{Z}_m = \{0, 1, \dots m-1\}$ and the vertices of $P_n$ with $\mathbb{Z}_n = \{0,1,\dots,n-1\}$ then $V(G_{m,n}) = \mathbb{Z}_m \times \mathbb{Z}_n$ and two vertices $(i,j)$ and $(i',j')$ are adjacent if and only if $|i-i'| + |j-j'| = 1$. We can thus represent $G_{m,n}$ in the two-dimensional coordinate system as in Figure~\ref{fig:grid_graph_4x4}.

For a grid graph $G$, we will denote by $L_G$ function assigning coordinate pairs to vertices of $G$. Sometimes we will write just $(i,j)$  thinking of vertex $u = L_G^{-1}((i,j))$.

We now introduce the class of \textit{toroidal grid graphs}. For positive integers $m \geq 3$ and $n \geq 3$, the \textit{toroidal grid graph} $TG_{m,n}$ is the Cartesian product of the cycles $C_m$ and $C_n$. When both $m$ and $n$ are even, $TG_{m,n}$ is bipartite. For the purpose of establishing an upper bound on the readability of grid graphs, we will consider the following special case of toroidal grid graphs. For a positive integer $n$, we introduce the graph $TG_n$ by setting $TG_n = TG_{4n,4n}$.

\begin{lemma}\label{lemma:lbcgn}
For every positive integer $n$, the graph $TG_n$ is bipartite and $r(TG_n) \geq 3$.
\end{lemma}

\begin{proof}
Since any grid graph is bipartite it is enough to show that vertices denoted with $(0,j)$ and $(4n-1,j)$ (as well as $(j,0)$ and $(j,4n-1)$ for $j = 0,1,\dots,4n-1$) are in different parts of bipartition of the grid graph $G_{4n,4n}$ because vertices denoted with $(0,j)$ and $(4n-1,j)$ ($(j,0)$ and $(j,4n-1)$) are not adjacent in $G_{4n,4n}$. This is true, since the distance between these two vertices in graph $G_{4n,4n}$ is $4n-1$, an odd number (recall that two vertices of a connected bipartite graph belong to the same part of the bipartition if and only if the distance between them is an even number). Thus, adding an edge between $(0,j)$ and $(4n-1,j)$ will not affect bipartiteness. We can use symmetry to infer the same for vertices $(j,0)$ and $(j,4n-1)$ for $j \in \{0,\dots, 4n-1\}$.

To show that $r(TG_n) \geq 3$, we will show that $TG_n$ contains $F$ (defined above) as induced subgraph. For example, take vertices denoted with $$(0,1), (1,0), (1,1), (1,2), (2,0), (2,1), (2,2).$$ Obviously, the subgraph of $TG_n$ induced by these vertices is isomorphic to $F$. By Lemma~\ref{thm:rdbind} and Lemma~\ref{lemma:rgd} we conclude $r(TG_n) \geq 3$.
\end{proof}

A bipartition of $TG_n = (S\cup T, E)$ can be obtained by the following rule: $S = \{(i,j) : i + j \equiv 0 \text{ (mod 2)} \}$ and $T = \{ (i,j) : i+ j \equiv 1 \text{ (mod 2)}\}$. In what follows, we will consider this bipartition of $TG_n$. For the sake of simplicity we  add $4n$ vertices (and corresponding edges) to the graph $TG_n$ with coordinates $(4n,0),(4n,1), \dots, (4n,4n-1)$, which we identify with $(0,0), (0,1), \dots, (0,4n-1)$, respectively, and we also add another $4n$ vertices with coordinates $(0,4n), (1,4n), \dots, (4n-1,4n)$, which we identify with $(0,0), (1,0), \dots, (4n-1,4n)$, respectively. We also add vertex $(4n,4n)$, which we identify with $(0,0)$. For example, vertices $(4n-1,0)$ and $(4n,0)$ are adjacent in $TG_n$ since $(0,0)$ and $(4n-1,0)$ are adjacent. Note that this identification is natural in view of the fact that the vertex set of $TG_n$ is $\mathbb{Z}_{4n}\times \mathbb{Z}_{4n}$. Figure~\ref{fig:tg_4x4} provides an example. Graph $TG_1$ is presented on the left image in Figure~\ref{fig:tg_4x4}. The edges which do not appear in corresponding grid graph ($G_{4,4}$) are marked with red. On the right image of Figure~\ref{fig:tg_4x4}, we added vertices $(j,4)$ and $(4,j)$ for $j\in\{0,\dots,4\}$ which are identified with vertices $(j,0)$ and $(0,j)$ for $j \in \{0,\dots,4\}$ respectively. For example, there is a red edge between $(0,4)$ and $(0,3)$ since there is a red edge between $(0,0)$ and $(0,3)$ (on the left image).
\bigskip
\bigskip
\begin{figure}[h!]
\begin{center}
\includegraphics[width=1\linewidth]{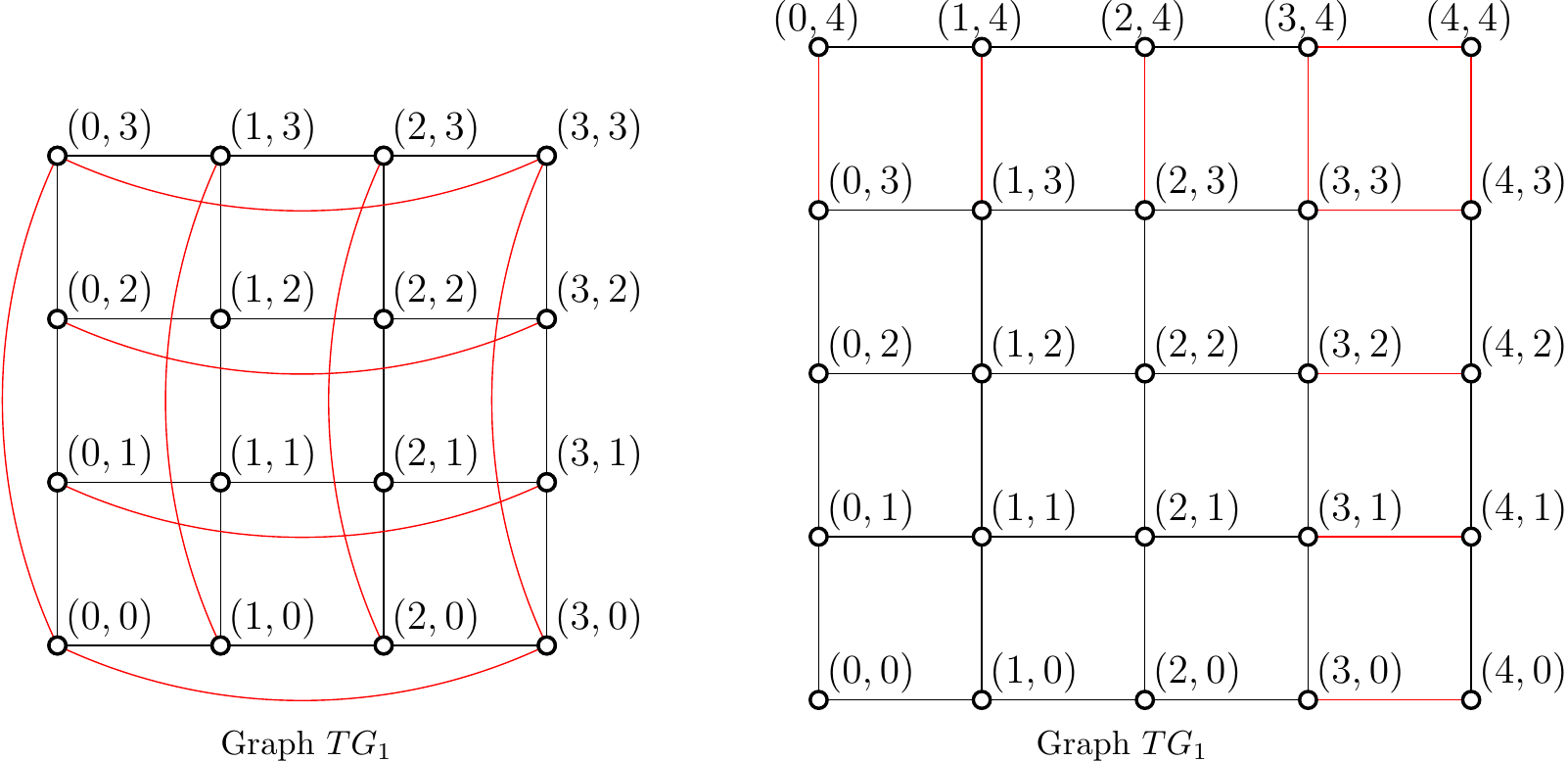}
\end{center}
\caption{Graph $TG_1$ presented in two equivalent but visually different ways.} \label{fig:tg_4x4}
\end{figure}

We will now prove that graph $TG_n$ is of readability exactly 3. We will define the length of the overlapping for each edge and then construct an overlap labeling. The pseudocode of the algorithm for determining the length of an overlapping of any two adjacent vertices of $TG_n$ is given in Algorithm~\ref{algorithm4}. In the algorithm, calculations are done modulo $4n$.
\begin{algorithm}[h!]\label{algorithm4}
\KwIn{Graph $TG_n$ and denotation of vertices $L = L_{TG_n}$ defined above}
\KwOut{$\{G_1, G_2, G_3\}$, $G_i$ for $i \in \{1,2,3\}$ is a spanning subgraph of $TG_n$ with the set of edges which will have an overlapping of size $i$ in the overlap labeling of $TG_n$}
\caption{Defining the size of the overlapping for edges of graph $TG_n$}
{
	$G_1 = G_2 = G_3 = (V(TG_n), \emptyset)$ \\
	\For{$i \in \{0,\dots, 4n-1\}$}
	{
		\For{$j \in \{0, \dots, 4n-1\}$}
		{
			\If{$(i \equiv 0 \pmod 2 \text{ and } j \equiv 0 \pmod 4) $ or \\
			$(i \equiv 1 \pmod 2 \text{ and } j \equiv 2 \pmod 4) $		
			}
			{
				$u_1 = (i,j)$;
				$u_2 = (i+1,j+1)$;
				$v_1 = (i,j+1)$;
				$v_2 = (i+1,j)$; \\
				add edges $u_1v_1$, $u_1v_2$, $u_2v_1$, $u_2v_2$ to $E(G_1)$
			}
			\If{$i \equiv 1 \pmod 2$}
			{
				$u = (i,j)$;
				$v = (i+1,j)$; \\
				add edge $uv$ to $E(G_2)$
			}
			\If{$(i \equiv 2 \pmod 4 \text{ and } j \equiv 0 \pmod 2)$ or \\
			$i \equiv 0 \pmod 4 \text{ and } j \equiv 1 \pmod 2)$			
			}
			{
				$u_1 = (i,j)$;
				$v_1 = (i,j+1)$;
				$u_2 = (i+1,j)$;
				$v_2 = (i+1,j+1)$; \\
				add edges $u_1v_1, u_2v_2$ to $E(G_3)$
			}
		}
	}
}
\Return{$\{ G_1, G_2, G_3 \}$;}
\end{algorithm}
The result of Algorithm~\ref{algorithm4} for graph $TG_2$ is given in Figure~\ref{fig:alg_exec2}.

The structure of graphs $G_1, G_2$ and $G_3$, obtained by Algorithm~\ref{algorithm4} with the input $(TG_n, L = L_{TG_n})$ for an arbitrary positive integer $n$, are captured by Theorem~\ref{thm:resultOfAlg}.

\begin{definition}\label{def:k_factor}
A \emph{regular graph} is a graph where each vertex has the same number of neighbors, i.e., every vertex has the same degree. A \emph{k-regular graph} is a regular graph with vertices of degree $k$. A \emph{k-factor} of a graph is spanning $k$-regular subgraph.
\end{definition}

\begin{theorem}\label{thm:resultOfAlg}
Let $TG_n$ and function $L_{TG_n}$ be given. Denote by $G_1, G_2, G_3$ the graphs obtained by the Algorithm~\ref{algorithm4} for the input $(TG_n, L = L_{TG_n})$. Then $G_1$ is a $2$-factor of $TG_n$ consisting of a disjoint union of 4-cycles, while $G_2$ and $G_3$ are $1$-factors of $TG_n$. Moreover, every edge of $TG_n$ is contained in a unique graph among $G_1$, $G_2$, $G_3$.
\end{theorem}

\begin{proof}
To each added 4-cycle $C$ in step 7, we can associate a unique vertex, denoted by $v(C)$, namely the vertex in $C$ closest to the origin (the vertex $u_1$ in line $6$ of the Algorithm~\ref{algorithm4}). Denote by $(i,j)$ coordinates of $v(C)$.
\begin{itemize}
\item If $v(C) \in S$ then $i \equiv 0 \pmod 4$ and $j \equiv 0 \pmod 2$.
\item If $v(C) \in T$ then $i \equiv 2 \pmod 4$ and $j \equiv 1 \pmod 2$.
\end{itemize}
It follows that any two $4$-cycles $C$ and $D$ added in line $7$ are vertex-disjoint.

We will now prove that $G_2$ is a $1$-factor of $TG_n$. The edges of $G_2$ are defined at step 10. Suppose there are vertices $u,v,v_1 \in V(G_2)$ such that $uv \in E(G_2)$ and $uv_1 \in E(G_2)$. Let $L_{TG_n}(u) = (i,j)$. Since the edges of $G_2$ are of the form $(i,j)(i+1,j)$ for some $i,j \in \{0,\dots,4n-1\}$, we have two possibilities for the coordinates of $v_1$ and $v_2$. They are $L_{TG_n}(v_1) = (i-1,j)$ and $L_{TG_n}(v_2) = (i+1,j)$ or $L_{TG_n}(v_2) = (i-1,j)$ and $L_{TG_n}(v_1) = (i+1,j)$. We may assume without loss of generality that $L_{TG_n}(v_1) = (i-1,j)$ and $L_{TG_n}(v_2) = (i+1,j)$. Then, since $uv_1 \in E(G_2)$ it must hold that $i \equiv 1 \text{ (mod 2)}$ and since $uv_2 \in E(G_2)$ we conclude $i \equiv 1 \text{ (mod 2)}$, a contradiction.

The proof that $G_3$ is $1$-factor of $TG_n$ is similar as above proof since all edges are of the form $(i,j)(i+1,j)$ for some $i \in \{0,\dots,4n-1\}$ and $j\in \{0,\dots,4n-1\}$.

The fact that the graphs $G_1, G_2$ and $G_3$ decompose $TG_n$ follows from the tables (Table~\ref{tbl:horizontal_decomposition} and Table~\ref{tbl:vertical_decomposition}) considering all possible types of edges in $TG_n$.

\begin{table}[h!]
\centering
\caption{Decomposition of the edges of $TG_n$ of the type $(i,j)(i,j+1)$. Each entry is one of the graphs $G_1$, $G_2$ or $G_3$, 
the one containing the edge $(i,j)(i,j+1)$.}
\label{tbl:horizontal_decomposition}
\begin{tabular}{|c|c|c|c|c|}
\hline
\diagbox[]{$j \mod 4$}{$i \mod 4$} & 0 & 1 & 2 & 3 \\ \hline
0 & $G_1$ & $G_1$ & $G_1$ & $G_1$ \\ \hline
1 & $G_2$ & $G_2$ & $G_2$ & $G_2$ \\ \hline
2 & $G_1$ & $G_1$ & $G_1$ & $G_1$ \\ \hline
3 & $G_2$ & $G_2$ & $G_2$ & $G_2$ \\ \hline
\end{tabular}
\end{table}

\begin{table}[h!]
\centering
\caption{Decomposition of the edges of $TG_n$ of the type $(i,j)(i+1,j)$. Each entry is one of the graphs $G_1$, $G_2$ or $G_3$,
the one containing the edge $(i,j)(i+1,j)$.}
\label{tbl:vertical_decomposition}
\begin{tabular}{|c|c|c|c|c|}
\hline
\diagbox[]{$j \mod 4$}{$i \mod 4$} & 0 & 1 & 2 & 3 \\ \hline
0 & $G_1$ & $G_3$ & $G_1$ & $G_3$ \\ \hline
1 & $G_3$ & $G_1$ & $G_3$ & $G_1$ \\ \hline
2 & $G_1$ & $G_3$ & $G_1$ & $G_3$ \\ \hline
3 & $G_3$ & $G_1$ & $G_3$ & $G_1$ \\ \hline
\end{tabular}
\end{table}

\end{proof}


Consider again the graph $TG_n$ for an arbitrary positive integer $n$ and denote by $S$ and $T$ the two parts of the bipartition of $TG_n$.
Let $G_1, G_2, G_3$ be the graphs obtained by Algorithm~\ref{algorithm4} for the input $(TG_n, L_{TG_n})$. To construct an overlap labeling of $TG_n$, we first assign a string of length 3 consisting of null characters to each vertex $v \in V(TG_n)$ and then we do the following:
\begin{enumerate}\label{proc:labproc}
\item For each 4-cycle $C$ of $G_1$ assign character $C$ to all $v \in V(C)$ as follows: if $v \in S$ then place the character to the third position of the string assigned to $v$, otherwise (when $v \in T$) place the character to the first position of the string assigned to $v$.
\item For each edge of $G_2$ construct an overlap of length $2$.
\item For each edge of $G_3$ construct an overlap of length $3$.
\end{enumerate}

We will now give arguments showing that the above procedure is well defined (see also example in Section 6.2). Denote by $s_u$ the string assigned to vertex $u$ and the $i$th character for $i \in \{1,2,3\}$ by $s_u(i)$. The first step is obviously well defined since by above $G_1$ is a disjoint union of $4$-cycles. After the first step, each vertex of $TG_n$ will be labeled with a character.

Let $e = uv \in E(G_2)$ and assume without loss of generality that $u \in S$ and $v \in T$. After the first step $s_u(3) = a$ and $s_v(1) = b$ for some distinct characters $a$ and $b$ (since by the construction of the algorithm, if $uv \in E(G_2)$ then it can not happen that $uv \in E(G_1)$). To make an overlap of length $2$ we set $s_u(2) = b$ and $s_v(2) = a$.

The third step is the most difficult one. After the second step, for each $u \in V(TG_n)$ the string $s_u$ has exactly two characters different from the null character (since $G_2$ is a $1$-factor of $TG_n$).
We will show that if there is an edge $e = uv \in E(G_3)$ and both $s_u$ and $s_v$ have two characters different from the null character (i.e., $s_u(i) \neq *$ for $i \in \{2,3\}$ and $s_v(j) \neq *$ for $j \in \{1,2\}$), then $s_u(2) = s_v(2)$.
Let $uv \in E(G_3)$ and $u \in S$, $v \in T$.
By above, there exists $e_1 \in E(G_2)$ such that $u$ is an endpoint of $e_1$ and there exists $e_2 \in E(G_2)$ such that $v$ is an endpoint of $e_2$.
In that case let $u_1$ be vertex of $V(TG_n)$ such that $uu_1 \in E(G_2)$ and $v_1 \in V(TG_n)$ such that $vv_1 \in E(G_2)$.
Then by the algorithm, it must hold $u_1v_1 \in E(G_1)$. Thus $s_{u_1}(1) = s_{v_1}(3) = a$ for some character $a$. Now, since $uu_1 \in E(G_2)$, by the above, we have $s_u(2) = s_{u_1}(1) = a$ and since $vv_1 \in E(G_2)$ then $s_v(2) = s_v(3) = a$ which implies $s_u(2) = s_v(2)$. Ta make an overlapping of length 3, we set $s_u(1) = c$ and $s_v(3) = b$ (see Section 7.1 for an example).

By the above arguments, we have an assignment to vertices of $TG_n$ such that if $e = uv \in E(TG_n)$ we have $ov(s_u,s_v) > 0$. To prove that such assignment is an overlap labeling of $TG_n$ we need to prove that if for some distinct vertices $u \in S$, $v \in T$ we have $ov(s_u, s_v) > 0$ then $uv \in E(TG_n)$. We will analyze only one case since we can use symmetry to state the same arguments for the other cases.

Suppose that for some $u \in S$ and $v \in T$, $ov(s_u,s_v) > 0$. If $ov(s_u, s_v) = 1$ then $s_u(1) = s_v(3)$. By our procedure we assign the same character to $s_u(1)$ and $s_v(3)$ if and only if $u$ and $v$ belong to the same $4$-cycle of $G_1$, thus $uv \in E(G_1)$ which imply $uv \in E(TG_n)$.

Suppose that $ov(s_u, s_v) = 2$. Let $s_u(3) = a$ and $s_v(1) = b$ for some characters $a$ and $b$.
Since $ov(s_u, s_v) = 2$ we have $s_u(2) = b$ and $s_v(2) = a$.
By our procedure such assignment to $s_u$ and $s_v$ can happen if there is edge $uv \in E(G_2)$ which implies $uv \in E(TG_n)$.

Lastly, suppose that $ov(s_u,s_v) = 3$ and let $s_u(1) = a$, $s_u(2) = b$ and $s_u(3) = c$ for some characters $a,b,c$. Note that $a \neq c$ since $ov(s_u, s_v) > 1$.
There exists a unique vertex $u_1 \in V(TG_n)$ such that $uu_1 \in E(G_2)$ and $s_{u_1}(1) = b$ because by the procedure we change the second character if and only if such a situation occurs. Similarly, there exists a unique $v_1 \in V(TG_n)$ such that $vv_1 \in E(G_2)$ and $s_{v_1}(3) = b$.
We conclude that $u_1$ and $v_1$ belong to the same $4$-cycle of $G_1$. Let $L_{TG_n}(u) = (i,j)$.
By symmetry, we may assume that vertices which belong to the same $4$-cycle in $G_1$ as vertex $u$ are denoted with $(i-1,j), (i, j+1), (i-1,j+1)$, the other cases will follow by symmetry (see Figure~\ref{fig:last_proof}).
Also, since $uu_1 \in E(G_2)$ by the construction of graph $TG_n$ and above assumptions it must follow that $LG_{TG_n}(u_1) = (i+1,j)$.
Since $u_1$ and $v_1$ belong to the same $4$-cycle of $G_1$, we have several possibilities: (i) $L_{TG_n}(v_1) = (i+2,j-1)$, (ii) $L_{TG_n}(v_1) = (i+2,j)$ and (iii) $L_{TG_n}(v_1) = (i+1,j-1)$. Situation (i) can not occur since in that case $v_1 \in T$ and $v \in T$ so there is no edge between them (see Figure~\ref{fig:last_proof}, left image).
If $L_{TG_n}(v_1) = (i+2,j)$ then by Algorithm~\ref{algorithm4} and since $vv_1 \in E(G_2)$, it must follow that $L_{TG_n}(v) = (i+3,j)$. By the algorithm, if $u$ and $v$ belong to different $4$-cycles of $G_1$ and $s_u(2) = s_v(2)$, then we can have an overlapping of length 3 if there is an edge between $u$ and $v$ (more precisely, if there is an edge $uv \in E(G_3)$ and we set $s_v(3) = s_u(3)$ and $s_u(1) = s_v(1)$). But, by the construction of the graph $TG_n$, vertices $(i,j)$ and $(i+3,j)$ are not adjacent (since $a \neq c$) , which means that $L_{TG_n}(v_1)$ can not be equal to $(i+2,j)$ (see Figure~\ref{fig:last_proof}, right image).

Lastly if $L_{TG_n}(v_1) = (i+1,j-1)$ then by the algorithm (since $vv_1 \in E(G_2)$) the only possibility is that $L_{TG_n}(v) = (i,j-1)$. But then, by the construction of the graph $TG_n$ $u$ and $v$ are adjacent and we obtained desired (see Figure~\ref{fig:last_proof_last_case}).

\begin{figure}[h!]
\begin{center}
\includegraphics[width=1\linewidth]{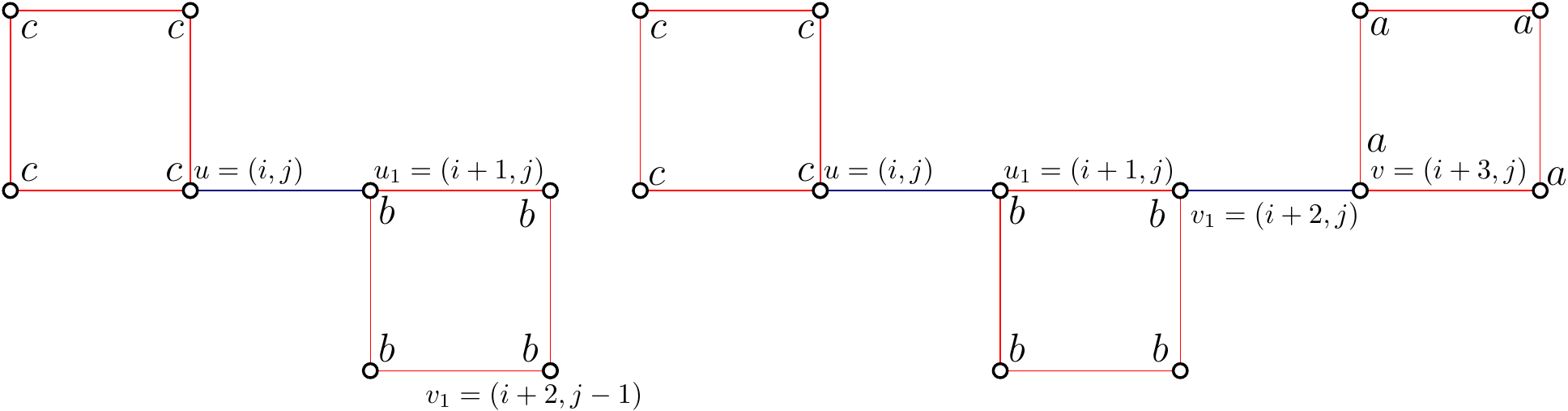}
\end{center}
\caption{The left image corresponds to the case $v_1 = L_{TG_n}^{-1}(i+2,j-1)$. The right image corresponds to the case $v = L_{TG_n}^{-1}(i,j+3)$.} \label{fig:last_proof}
\end{figure}

\begin{figure}[h!]
\begin{center}
\includegraphics[width=0.3\linewidth]{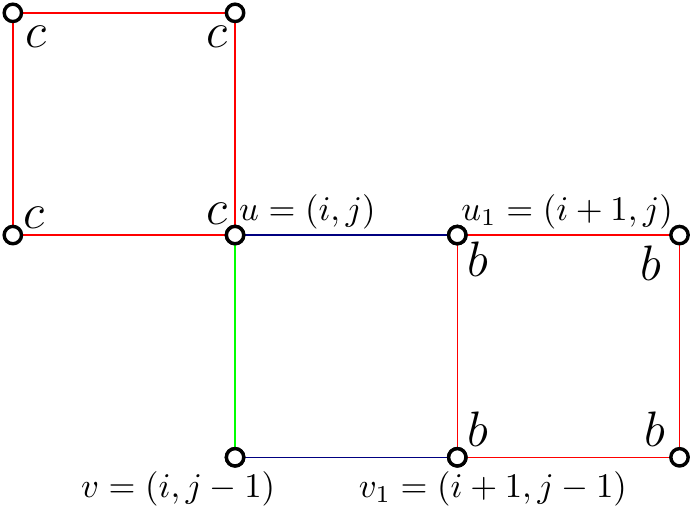}
\end{center}
\caption{The third case of above proof: if $v_1 = (i+1,j-1)$ and $vv_1 \in E(G_2)$ then the only possibility is $v = (i,j-1)$. In that case $uv \in E(TG_n)$.} \label{fig:last_proof_last_case}
\end{figure}

There are three further cases with respect to the position of $u$ in the $4$-cycle of $G_1$. In all those cases we apply similar arguments to the ones given above.

By the stated arguments, the assignment obtained by above procedure is an overlap labeling of length $3$. Combining this result with Lemma~\ref{lemma:lbcgn} we get that $r(TG_n) = 3$. Now we are able to prove Theorem~\ref{thm:rgrid}:
\begin{proof}[Proof of Theorem~\ref{thm:rgrid}]
Let $G_{m,n}$ be grid graph of size $m\times n$, $m\geq 3$ and $n\geq 3$.
One can easily check that the subgraph of $G_{m,n}$ induced by vertices $$(0,1), (1,0), (1,1), (1,2), (2,0), (2,1), (2,2)$$ is isomorphic to graph $F$.
By Theorem~\ref{thm:rdbind} and Lemma~\ref{lemma:rgd} we obtain $r(G_{m,n}) \geq 3$. Let $k = \max\{m,n\}+1$. Clearly, $G_{m,n}$ is induced subgraph of $TG_{k}$ and hence $r(G_{m,n}) \leq r(TG_{k})$. By the above discussion, $r(TG_{k}) = 3$. Thus $$3 = r(F) \leq r(G_{m,n}) \leq r(TG_{k}) =3$$.
\end{proof}

\section{An example}
We conclude the chapter with a concrete example illustrating the construction of an overlap labeling of a toroidal grid graph $TG_2$ as described above (see Figure~\ref{fig:alg_exec2}).

Let $G_1,G_2$ and $G_3$ be the graphs obtained by Algorithm~\ref{algorithm4} for input $(TG_2,L=L_{TG_2})$. $G_1$ is a disjoint union of 4-cycles denoted by red color, $G_2$ consists of edges colored with blue color and $G_3$ consists of edges colored with green color. By the procedure for constructing an optimal overlap labeling, we first assign a unique character to each 4-cycle (note that the size of the alphabet used to construct an overlap labeling of length 3 is equal to number of 4-cycles in $G_1$). Suppose for example that we assign characters $a,b,c$ to cycles with lower left vertices $(0,0),(2,1),(2,7)$, respectively. Since $(1,1)(2,1) \in E(G_2)$ and $(1,0)(2,0) \in E(G_2)$ we assign strings $*ba$, $ba*$, $*ac$ to vertices $(1,1), (2,1), (2,0)$, respectively (recall that $*$ stands for null character, i.e., we did not yet assign any character on the position where $*$ is placed). Since $(2,0)(2,1) \in E(G_3)$ we have to construct overlapping of length 3. This is possible since the second character of $(2,1)$ is equal to the second character of $(2,0)$. We assign string $bac$ to both vertices $(2,0)$ and $(2,1)$ and make overlapping of length 3 without affecting to overlaps constructed while processing edges of $G_2$. The other edges are processed similarly.

\bigskip
\begin{figure}[h!]
\begin{center}
\includegraphics[width=0.6\linewidth]{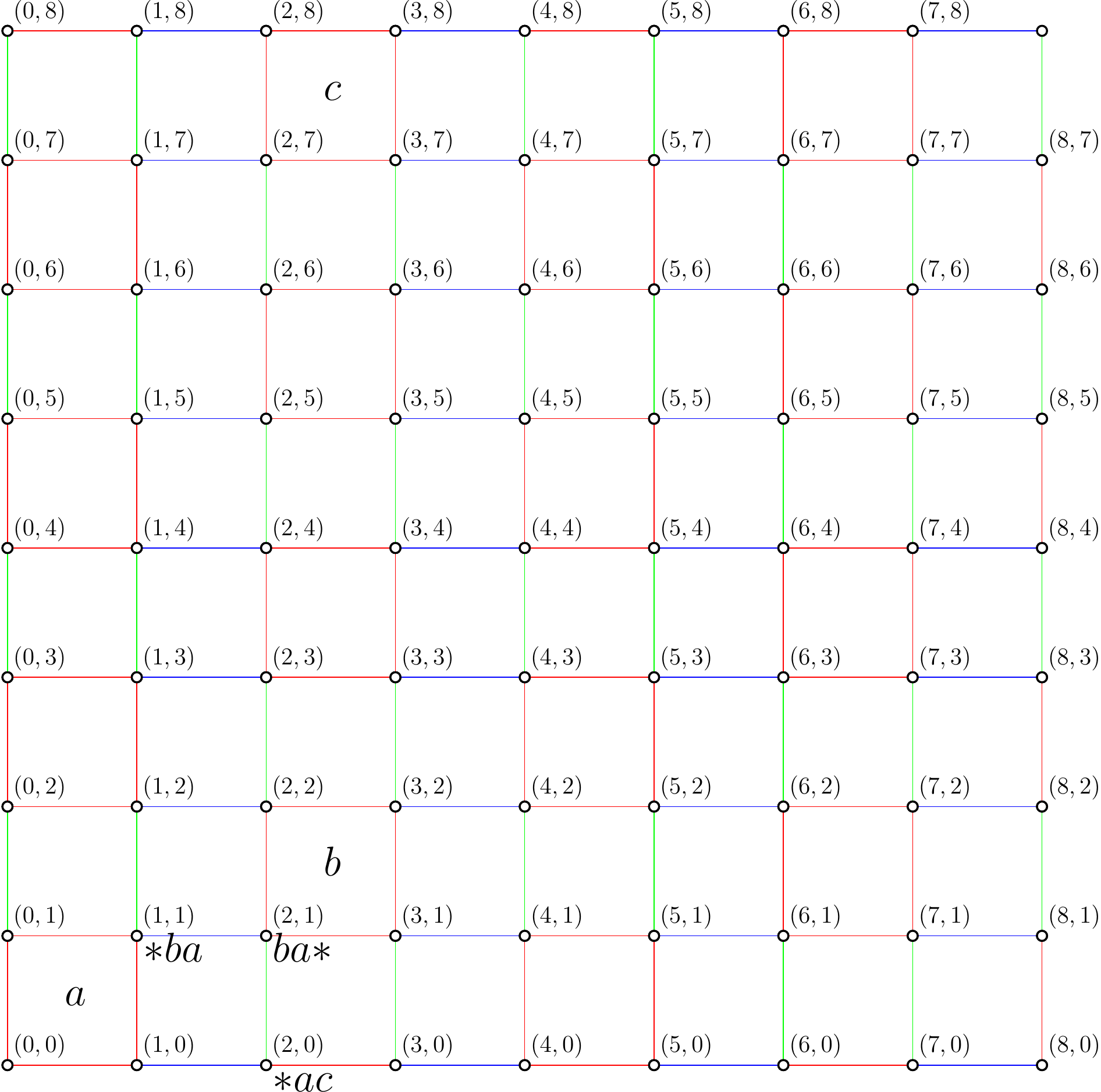}
\end{center}
\caption{Constructing an overlap labeling of $TG_2$} \label{fig:alg_exec2}
\end{figure}


\chapter{Conclusion}
\thispagestyle{fancy}

In the final project paper, we reviewed an algorithm for building a set of strings given its overlap graph given by Braga and Meidanis \cite{BragaMeidanis}. Although the derived running time of the algorithm is exponential in the maximum degree, one can make it polynomial using programming languages which allows memory manipulation (e.g., using pointers) such as \textit{C++}. However, the upper bound for the readability derived with Algorithm~\ref{algorithm11} is in general very weak. This motivated us to formulate an ILP model for the exact computation of readability (for balanced bipartite graphs as well as for digraphs).While polynomial, the number of variables and constraints of the ILPs can be large already for moderately sized graphs. However, with fast ILP solvers, the models can be used to understand the behavior of readability of small graphs and to state hypotheses for readability in general. We implemented the ILP for bipartite graphs in Java~\cite{java} using IBM$^{\scriptscriptstyle\textregistered{}}$ ILOG$^{\scriptscriptstyle\textregistered{}}$ CPLEX$^{\scriptscriptstyle\textregistered{}}$ Optimization Studio~\cite{cplex}. The implementation is available at \cite{github}.

We reviewed bounds for readability using tools from graph theory developed in \cite{MilanicRayanMedvedev} and presented theorems that are the basis for the derivation of polynomial time algorithms for checking if a graph is of readability at most $2$. The lower bound given by the parameter \textit{distinctness} can be computed in polynomial time. The time complexity of computing bounds given by \textit{HUB}-number of a graph is unknown and such bounds can therefore only be used in theory. However, this theory motivated us to establish the readability of grid graphs and toroidal grid graphs. We showed that readability of such graphs is at most $3$ and moreover for any (toroidal) grid graph of size $m\times n$ for some integers $m \geq 3,n \geq 3$, the readability is exactly 3.
We also presented a polynomial time algorithm for constructing an optimal overlap labeling of such graphs.
It can be verified that the HUB-number of graph $F$ (see Figure~\ref{fig:graph_F}) is at least 3, which implies that for all large enough grid graphs and for all toroidal grid graphs, the readability coincides with their HUB-number (which in general is only a lower bound for the readability).
In future work, we will try to identify further graph classes satisfying this condition and develop an explicit formula for an optimal overlap labeling of toroidal grid graphs.


\chapter{Povzetek naloge v slovenskem jeziku}
\thispagestyle{fancy}

V zaklju\v cni nalogi smo obravnavali \textit{od\v citljivost} grafa. Naj bo $C$ kon\v cna mno\v zica besed nad poljubno kon\v cno abecedo. Naj bosta $s_1$ in $s_2$ dve poljubni besedi iz mno\v zice $C$. Pravimo, da obstaja prekrivanje velikosti $k$ med besedama $s_1$ in $s_2$, \v ce velja $\suff(s_1,k) = \pref(s_2,k)$, kjer $\suff(s,k)$ ozna\v cuje pripono besede $s$ dol\v zine $k$, $\pref(s,k)$ pa njeno predpono dol\v zine $k$. Na podlagi tega definirajmo funkcijo $ov : C\times C \to \mathbb{N}\cup \{0\}$ s predpisom $ov(s_1,s_2) = \min_{k>0}\{\suff(s_1,k)=\pref(s_2,k)\}$, \v ce neko tako prekrivanje obstaja, in $0$, sicer. Z besedami povedano, $ov(s_1,s_2)$ ozna\v cuje velikost najmanj\v sega prekrivanja med besedama $s_1$ in $s_2$.
Iz tega lahko tvorimo digraf na naslednji na\v cin: vsaka beseda predstavlja eno vozli\v s\v ce in dve vozli\v s\v ci sta povezani natanko tedaj ko je $ov(s_1,s_2) > 0$, kjer sta $s_1$ in $s_2$ besede, ki ustrezata vozli\v s\v cema.
Graf, pridobljen na ta na\v cin, bomo imenovali \textit{graf prekrivanj} dane mno\v zice besed. Poglejmo obratni problem.
 Naj bo dan digraf $D = (V,A)$. Naloga je poiskati tako mno\v zico besed $C$, da bo njen graf prekrivanj izomorfen danemu digrafu $D$.
 Funkcijo, ki vsakem vozli\v s\v cu digrafa $D$ dodeli neko besedo, bomo imenovali \textit{ozna\v cevalna funkcija}.
 Dol\v zina ozna\v cevalne funkcije je najve\v cja dol\v zina besede med vsemi besedami, ki jih ozna\v cevalna funkcija dodeli vozli\v s\v cem. Od\v citljivost digrafa $D$ definiramo kot tako najmanj\v se pozitivno \v stevilo $k$, da obstaja injektivna ozna\v cevalna funkcija digrafa $D$ dol\v zine $k$.
 Algoritem~\ref{algorithm11} (na strani \pageref{algorithm11}) poka\v ze, da je od\v citljivost dobro definiran parameter. Iz Algoritma~\ref{algorithm11} takoj dobimo zgornjo mejo za od\v citljivost, in sicer $2^{p+1}-1$, kjer je $p=\max\{\Delta^+(D), \Delta^-(D)\}$. Ra\v cunska zahtevnost problema izra\v cuna od\v citljivosti danega digrafa \v se ni znana.

\begin{sloppypar}
\v Ceprav je od\v citljivost v osnovi definirana za digrafe, s pomo\v cjo Izreka~\ref{thm:rdb_assymp} (na strani \pageref{thm:rdb_assymp}) lahko \v studiramo od\v citljivost na dvodelnih uravnote\v zenih grafih. Definicija od\v citljivosti za dvodelne (uravnote\v zene) grafe je skoraj ista kot za digrafe. Razlika je zgolj v tem, da ne zahtevamo injektivnosti ozna\v cevalne funkcije.

Od\v citljivost grafa je te\v zko pora\v cunati \v ze za majhne grafe. V ta namen smo v \v cetrtem poglavju predstavili celo\v stevilski linearni program (CLP) za natan\v cno ra\v cunanje od\v citljivosti dvodelnih uravnote\v zenih grafov. V petem poglavju smo predstavili CLP za izra\v cun od\v citljivosti danega digrafa. V praksi je opisani CLP uporaben zgolj za majhne grafe, ker je sicer \v stevilo spremenljivk in omejitev zelo veliko in izvajanje programa (implementiranega s pomo\v cjo CPLEXa) traja zelo dolgo.
\end{sloppypar}
Poleg zgornje meje dobljene z Algoritmom~\ref{algorithm11} je predstavljenih \v se nekaj drugih.
Obravnavan je polinomski algoritem za prepoznavanje grafov od\v citljivosti 2.
Potem smo predstavili spodnjo mejo s pomo\v cjo parametra \textit{razlo\v cljivosti}. Spodnjo in zgornjo mejo lahko pora\v cunamo tudi s pomo\v cjo HUB-\v stevila ($hub(G)$) danega grafa. Velja, da je od\v citljivost dvodelnega uravnote\v zenega grafa ve\v cja ali enaka HUB-\v stevilu in je kve\v cjemu $2^{hub(G)}-1$. Ra\v cunska zahtevnost izra\v cuna HUB-\v stevila danega dvodelnega grafa \v zal ni znana, tako da so meje na osnovi HUB-\v stevila uporabne zgolj v teoriji.

V zadnjem poglavju zaklju\v cne naloge smo obravnavali od\v citljivost dvodimenzionalnih in toroidalnih mre\v z. Pokazali smo, da je od\v citljivost teh kve\v cjemu 3.
Natan\v cneje, od\v citljivost dane mre\v ze ali toroidalne mre\v ze je enaka 3 (razen za $2\times n$ mre\v ze, ki so od\v citljivosti kve\v cjemu 2).
Pri tem smo \v se predstavili algoritem, ki v polinomskem \v casu zgradi optimalno ozna\v cevalno funkcijo dol\v zine 3 za dano toroidalno mre\v zo.


\newpage


%
%
%
%
\end{document}